\documentclass[prx,a4paper,aps,twocolumn,superscriptaddress,nofootinbib,longbibliography,10pt]{revtex4-1}

\usepackage{amsmath,amsthm,amsfonts,amssymb,mathrsfs,mathdots,graphicx,xcolor,times,xfrac,mathtools,enumitem,xr,subfigure,bbm,verbatim,comment,dsfont,soul,array,multirow}

\usepackage[breaklinks=true]{hyperref}
 \hypersetup{
   colorlinks   = true,
   urlcolor     = blue,
   linkcolor    = blue,
   citecolor   = red
}
\usepackage{cleveref}
\crefformat{equation}{Eq.~(#2#1#3)}
\crefformat{figure}{Fig.~#2#1#3}
\crefrangeformat{equation}{Eqs.~#3(#1)#4--#5(#2)#6}
\crefname{pleqs}{Equations}{Equations}
\crefformat{pleqs}{Eqs.~(#2#1#3)}

\usepackage{tikz}
\usetikzlibrary{shapes,backgrounds,fit,decorations.pathreplacing,arrows,decorations.markings}
\tikzstyle{vecArrow} = [thick, decoration={markings,mark=at position
   1 with {\arrow[semithick]{open triangle 60}}},
   double distance=1.4pt, shorten >= 5.5pt,
   preaction = {decorate},
   postaction = {draw,line width=1.4pt, white,shorten >= 9.5pt}]
\usepackage{verbatim}

\usepackage[utf8]{inputenc} 
\usepackage{bbm}
\usepackage{bm}
\usepackage{amsbsy}
\usepackage{amsthm}
\usepackage{amssymb}
\usepackage{amsfonts}
\usepackage{amsmath}
\usepackage{dsfont} % for symbols like the identity: \mathds{1}
\usepackage{graphicx} % for graphics
\usepackage{epsfig}
\usepackage{epstopdf}
\usepackage{dsfont}
\usepackage{multibib}
\usepackage{xcolor}
\usepackage{MnSymbol}
\usepackage{enumerate}%allows different styles of enumerate environment
\usepackage{float}
\definecolor{darkgreen}{RGB}{50,190,50}
\definecolor{darkblue}{RGB}{0,0,190}
\definecolor{darkred}{RGB}{238,0,0}
\definecolor{quantum}{RGB}{83,37,127}
\definecolor{quantumlight}{RGB}{169,146,191}
\usepackage{soul}
\newcommand{\pr}{^{\prime}}

\newcommand{\ket}[1]{\ensuremath{\left|\right.\!{#1}\!\left.\right\rangle}}

\newcommand{\bra}[1]{\ensuremath{\left\langle\right.\!{#1}\!\left.\right|}}
\newcommand{\braket}[2]{\ensuremath{\langle{#1}|{#2}\rangle}}
\newcommand{\ketbra}[2]{\ensuremath{|{#1}\rangle\!\langle{#2}|}}

\newcommand{\subtiny}[3]{\ensuremath{_{\hspace{#1 pt}\protect\raisebox{#2 pt}{\tiny{$ #3$}}}}}
\newcommand{\suptiny}[3]{\ensuremath{^{\hspace{#1 pt}\protect\raisebox{#2 pt}{\tiny{$ #3$}}}}}

\newcommand\scalemath[2]{\scalebox{#1}{\mbox{\ensuremath{\displaystyle #2}}}}

\DeclareMathOperator{\diag}{diag}

\newcommand{\djj}{d\kern-0.4em\char"16\kern-0.1em}
\newtheorem{theorem}{Theorem}[]
\newtheorem{definition}{Definition}

\newtheorem{coro}{Corollary}[section]
\newtheorem{prop}{Proposition}[section]
\newtheorem{teo}{Theorem}[section]

\newtheorem{defi}{Definition}
\numberwithin{defi}{section}

\renewcommand{\thesection}{\arabic{section}}
\renewcommand{\thesubsection}{\arabic{section}.\Alph{subsection}}
\renewcommand{\thesubsubsection}{\arabic{section}.\Alph{subsection}.\arabic{subsubsection}}
\makeatletter
\renewcommand{\p@subsection}{}
\renewcommand{\p@subsubsection}{}
\makeatother

\newcommand{\M}[1]{\mathcal{#1}}

\usepackage{physics}
\usepackage[customcolors]{hf-tikz}
\tikzset{style green/.style={
    set fill color=green!50!lime!60,
    set border color=white,
  },
  style cyan/.style={
    set fill color=cyan!90!blue!60,
    set border color=white,
  },
  style orange/.style={
    set fill color=orange!80!red!60,
    set border color=white,
  },
  style hordash/.style={
    set fill color=white,
    set border color=black,
  },
    style white/.style={
    set fill color=white,
    set border color=white,
  },
     style rose/.style={
    set fill color= magenta!70!pink!70, %pink!20!blue!40!,
    set border color=white,
  },
    style redash/.style={
    set fill color=white,
    set border color=red,
  },
  hor/.style={
    above left offset={-0.09,0.25},
    below right offset={0.09,-0.05},
    #1
  },
  ver/.style={
    above left offset={-0.09,0.35},
    below right offset={0.09,-0.1},
    #1
  }
}
\begin{document}

\title{Fermionic anyons: entanglement and quantum computation from a resource-theoretic perspective}

%%% ordem alfábetica no primeiro nome. Allan e Antônio como primeiro autores 
\author{Allan Tosta}
\thanks{A. T. and A. C. L. contributed equally to this work.}
\affiliation{Universidade Federal do Rio de Janeiro, Caixa Postal 68528, Rio de Janeiro, RJ 21941-972, Brazil}
\affiliation{Quantum Research Center, Technology Innovation Institute, Abu Dhabi, UAE}

\author{Ant\^onio C. Louren\c{c}o }
\thanks{A. T. and A. C. L. contributed equally to this work.}
\affiliation{Department of Physics and Astronomy, University of Iowa, Iowa City, Iowa 52242, USA}

\author{Daniel Brod}
\affiliation{Instituto de F\'isica, Universidade Federal Fluminense, 24210-346 Niter\'oi, Brazil}

\author{Fernando Iemini}
\affiliation{Instituto de F\'isica, Universidade Federal Fluminense, 24210-346 Niter\'oi, Brazil}
\affiliation{Institute for Physics, Johannes Gutenberg University Mainz, D-55099 Mainz, Germany}

\author{Tiago Debarba}
\email{debarba@utfpr.edu.br}
\affiliation{Departamento Acad{\^ e}mico de Ci{\^ e}ncias da Natureza, Universidade Tecnol{\'o}gica Federal do Paran{\'a} (UTFPR), Campus Corn{\'e}lio Proc{\'o}pio, Avenida Alberto Carazzai 1640, Corn{\'e}lio Proc{\'o}pio, Paran{\'a} 86300-000, Brazil}
\affiliation{Atominstitut, Technische Universit{\"a}t Wien, Stadionallee 2, 1020 Vienna, Austria}

\date{\today}
\begin{abstract}
Quantum computational models can be approached via the lens of resources needed to perform computational tasks, where a computational advantage is achieved by consuming specific forms of quantum resources, or, conversely, resource-free computations are classically simulable. Can we similarly identify quantum computational resources in the setting of more general quasi-particle statistics? In this work, we develop a framework to characterize the separability of a specific type of one-dimensional quasiparticle known as a fermionic anyon. As we evince, the usual notion of partial trace fails in this scenario, so we build the notion of separability through a fractional Jordan-Wigner transformation, leading to an entanglement description of fermionic-anyon states. We apply this notion of fermionic-anyon separability, and the unitary operations that preserve it, mapping it to the free resources of matchgate circuits. We also identify how entanglement between two qubits encoded in a dual-rail manner, as standard for matchgate circuits, corresponds to the notion of entanglement between fermionic anyons.
\end{abstract}
\maketitle

\emph{Introduction.} Over the last five decades, our notion of identical particles in nature has expanded beyond fermions and bosons. Many two-dimensional systems were shown to contain anyonic excitations \cite{wilczek_magnetic_1982,nayak_non-abelian_2008}, which are quasi-particles characterised by the non-trivial phases their wave functions acquire under particle exchange. These include fractional quantum Hall states \cite{arovas_fractional_1984,stern_anyons_2008}, topological spin liquids \cite{savary_quantum_2016,zhou_quantum_2017}, and semiconductor nanowire arrays \cite{stanescu_majorana_2013,sarma_majorana_2015}. These systems are seen as possible platforms for fault-tolerant quantum computing \cite{wilczek_magnetic_1982,sarma_majorana_2015}, given their inherent error-correcting properties \cite{fowler_surface_2012,brell_generalized_2015,litinski_lattice_2018} and recent experimental evidence of their existence and predicted properties \cite{nakamura_direct_2020}. 

{Although anyons are most commonly associated with two-dimensional systems, they can also be defined in one dimension.} Some notable examples are anyons obtained by dimensional reduction \cite{hansson_anyons_1991,ha_fractional_1995}, or appearing as a free-particle description of one-dimensional systems with two-body interactions \cite{lieb_exact_1963,calogero_ground_1969,haldane_model_1988,shastry_exact_1988,olshanetsky_quantum_1983}. {The one dimensional anyons considered here} are motivated by their role in solving many-body systems with three-body interactions \cite{kundu_exact_1999,batchelor_one-dimensional_2006,batchelor_fermionization_2006,calabrese_correlation_2007,patu_correlation_2007,hao_ground-state_2008,keilmann_statistically_2011} and have been investigated in optical lattice implementations \cite{greschner_density-dependent_2014,cardarelli_engineering_2016,liang_floquet_2018,schweizer_floquet_2019}. Although they lack the topological properties of their two-dimensional counterparts \cite{harshman_anyons_2020}, their relation to standard fermionic and bosonic systems via generalised Jordan-Wigner transformations \cite{meljanac_r_1994,doresic_generalized_1994,meljanac_unified_1996} makes them a good case study for generalisations of quantum computing with bosonic and fermionic linear optics.

In this letter, we develop a framework to define and investigate the separability of fermionic anyons. Since this is well-understood for fermions, the naive approach is to directly repurpose definitions of fermionic entanglement to their anyonic counterparts \cite{bonderson_anyonic_2017,mani_quantum_2020,zhang_statistically_2020,sreedhar_quantum_2022}. However, this sometimes leads to nonsensical results. {For example, we notice that {single-particle} transformations on a manifestly unentangled pure state can result in states with a nonzero entanglement entropy.}

{Within subspaces of fixed particle number, we circumvent these problems by a well-motivated approach for single-particle entanglement, and revise the definition of Schmidt coefficients of a composite fermionic-anyon state based on a non-canonical transformation over the anyonic states. Specifically, we map the anyonic algebra to another system that satisfies an anti-commutative algebra, and prove that the Schmidt
coefficients of the resulting mapped state coincide with those of the original anyonic state.}  

We showcase our approach by investigating the connection between separability and classical simulability in these systems. Free-fermionic quantum circuits  \cite{terhal_classical_2002} and matchgate computing \cite{valiant_quantum_2002,jozsa_matchgates_2008} are quantum computing settings where separability and computational power are tightly connected. Nearest-neighbor matchgate circuits can be mapped to free-fermion dynamics, and both are known to be classically simulable. However, it {is known} that supplementing these systems with any non-matchgate operation (in the fermionic picture, adding an interaction between particles) or any non-matchgate generated state (respectively, any non-Gaussian fermionic state) is enough to allow for universal quantum computation \cite{hebenstreit_all_2019}. Here, we leverage this connection to make a similar statement for fermionic anyons. In particular, we notice that both nonseparable states or transformations, as per our definition, can be seen as computational resources. Moreover, specific values of the fermionic-anyon phase $\varphi$ recover well-known results. Fermionic-anyon dynamics reduces to fermionic linear optics when $\varphi=0$, and to ``qubit linear optics'' (or matchgate quantum computing) when $\varphi=\pi$.\\

\emph{One-dimensional fermionic anyons.} Given a one-dimensional set of $m$ sites (or modes), we define a family of operator algebras $\{\mathcal{A}^{\varphi}_{m}|\text{ }\varphi\in[0,2\pi)\}$ over $\mathbb{C}$, generated by operators $\{a_{\varphi,i}|\text{ }i=1,...,m\}$, satisfying 
\begin{equation}\label{eq: anyon def}
\begin{split}
	&a_{\varphi,i}a_{\varphi,j}^{\dagger}+e^{-i\varphi\epsilon_{ij}}a_{\varphi,j}^{\dagger}a_{\varphi,i}=\delta_{ij}, \\
	&a_{\varphi,i}a_{\varphi,j}+e^{i\varphi\epsilon_{ij}}a_{\varphi,j}a_{\varphi,i}=0,
\end{split}
\end{equation}
with $\epsilon_{ij}$ given by
\begin{equation*}
    \epsilon_{ij}=\left\{\begin{matrix}
1, & \textnormal{if} \ \ i<j\\ 
0, & \textnormal{if} \ \ i=j\\ 
-1, & \textnormal{if} \ \ i>j.
\end{matrix}\right.
\end{equation*}

The variable $\varphi$ is called the \textit{statistical parameter} and determines the kind of particle described by the algebra. If $\varphi=0$ we identify $f_{i}:=a_{0,i}$, and $\mathcal{A}^{0}_{m}\equiv \mathcal{F}_{m}$, where $\mathcal{F}_{m}$ is the algebra of $m$-mode fermionic operators. If $\varphi=\pi$, then for all $i,j$ we have $[a_{\pi,i},a_{\pi,j}]=0$ as well as $\{a_{\pi,i},a^\dagger_{\pi,j}\}=0$, and we identify $\mathcal{A}^{\pi}_{m}\equiv\mathcal{Q}_{m}$, where $\mathcal{Q}_{m}$ is the algebra of operators for $m$-mode hardcore bosons, or qubits \cite{WuLydar2002}. For any other value of $\varphi$, the algebra $\mathcal{A}^{\varphi}_{m}$ describes particles with exotic exchange statistics called \textit{fermionic anyons}.

In the Supplemental Material \ref{app: anyon_description}, we show that, for all $\varphi$, the algebras $\mathcal{A}^{\varphi}_{m}$ have a well-defined Fock-space representation with number operators of the form $a^{\dagger}_{\varphi,i}a_{\varphi,i}$. Therefore, a general pure state of $N$ fermionic anyons has the form
\begin{equation}
 \ket{\psi} = \sum_{I_N} w_{I_N} a_{\varphi,i_1}^\dagger ... a_{\varphi,i_N}^\dagger \ket{\textrm{vac}},
\end{equation}
where $\ket{\textrm{vac}}$ is the vacuum state, $I_{N}=\{i_1<...<i_N\}$ is a shorthand for the list of particle indices, and 
$\sum_{I_N} |w_{I_N}|^2=1$.

\emph{Separability for fermionic anyons.} {For a quantum system with two sets of degrees of freedom, a standard quantifier of correlations for pure states is the entanglement entropy \cite{nielsen_quantum_2010}, 
\begin{equation}\label{eq.definition.ent.entropy}
 E(\ket{\psi}) = S( \rho_{\textrm{red}}).
\end{equation}
Here, $S(\rho)$ is the von Neumann entropy of $\rho$, and $\rho_{\textrm{red}}$ is the reduced state obtained by tracing out one of the subsystems.} 
{However, naively applying a particle partial trace on systems of fermionic anyons and computing its entanglement according to Eq.~\eqref{eq.definition.ent.entropy} can lead to nonsensical results.} {To illustrate that, consider the state 
$$\ket{\psi_{\theta}}= \frac{1}{\sqrt{2}}\left(a^{\dagger}_{\varphi,1}a_{\varphi,2}  + \cos\theta a^{\dagger}_{\varphi,1}a^{\dagger}_{\varphi,4} + i\sin\theta a^{\dagger}_{\varphi,2}a^{\dagger}_{\varphi,4}\right)\ket{\text{vac}}.$$ 
It can be obtained by applying a fermionic-anyon \emph{single-particle operation} on a manifestly separable state $\ket{\psi} = \frac{1}{\sqrt{2}}a^{\dagger}_{\varphi,1}(a^{\dagger}_{\varphi,2}+a^{\dagger}_{\varphi,4})\ket{\text{vac}}$, i.e.,
\begin{align}
    \ket{\psi_{\theta}} &= \exp[i\theta(a^{\dagger}_{\varphi,1}a_{\varphi,2}+a^{\dagger}_{\varphi,1}a_{\varphi,2})]\ket{\psi}.
\end{align}
Since particle entanglement is invariant under {single-particle} operations, the entanglement entropy of $\ket{\psi_{\theta}}$ should also be invariant, and hence zero by construction. As shown in Fig.~\ref{fig: entanglement entropy sp}, however, that is not the case. More details about this example can be found in Sec.~\ref{app: example} of Supplemental Material. 

For {standard} fermions, single-particle operations must act as changes of basis over single-particle systems, implying they have the second-quantized form 
\begin{equation}
{f}^{\dagger}_{i} \rightarrow Uf_{i}^{\dagger}U^{\dagger}=\sum_{j=1}^{m}U_{ij}f^{\dagger}_{j}, \label{eq:canonical}
\end{equation}
{where $U_{ij}$ are elements of an $m\cross m$ unitary matrix. This map is well-defined for fermions because it is canonical, i.e., does not change particle commutation relations.} However,  defining single-particle operations for fermionic anyons by analogy with \cref{eq:canonical} (i.e.\ replacing $f$ with $a_{\varphi}$) does not produce a canonical transformation. To properly define {these} operations for fermionic anyons, we must find an appropriate definition for their \emph{canonical transformations}.

%%%%%%%%%%%%%%%%%%%%%%%%%%%%%%%%%%%%%%%%%%%%%%%%%%
\begin{figure}[t]
    \centering
    \includegraphics[scale=0.57]{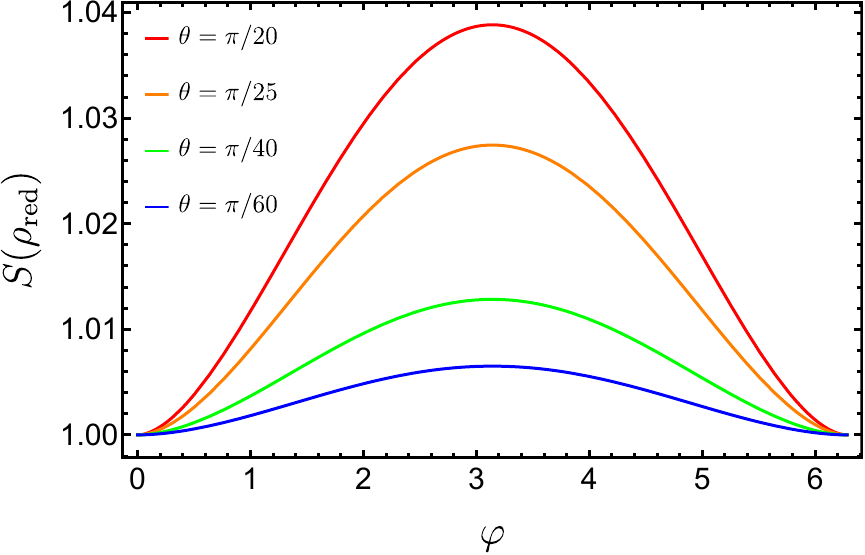}
    \caption{Fermionic anyon entanglement as defined in Eq.\eqref{eq.definition.ent.entropy} of the state $\ket{\psi\subtiny{0}{0}{\theta}} $. {The plot shows how the von Neumann entropy of the single particle state varies as a function of $\varphi$.}}
    \label{fig: entanglement entropy sp}
\end{figure}
%%%%%%%%%%%%%%%%%%%%%%%%%%%%%%%%%%%%%%%%%%%%%%%%%%%%

As shown in \cite{meljanac_generalized_1994,osterloh_fermionic_2000}, creation and annihilation operators for fermionic anyons (for any  $\varphi$) can be identified with operators in the usual fermionic algebra via the relation 
% escrever J_phi como sendo uma definicao
\begin{equation}\label{eq: jw-trad}
J_{\varphi}(a_{\varphi,j})=f_{j}e^{-i\varphi\sum^{j-1}_{k=1}f^{\dagger}_{k}f_{k}},%\quad J_{\varphi}(a^{\dagger}_{\varphi,j})=f^{\dagger}_{j}e^{i\varphi\sum^{j-1}_{k=1}f^{\dagger}_{k}f_{k}},
\end{equation} 
known as the \textit{fractional Jordan-Wigner transform} (JWT). It follows that  $a^{\dagger}_{\varphi,i}a_{\varphi,i}=f^{\dagger}_{i}f_{i}$, from which we obtain the inverse relationship
\begin{equation}\label{eq: jw-trad-inv}
J^{-1}_{\varphi}(f_{j})=a_{\varphi,j}e^{i\varphi\sum^{j-1}_{k=1}a^{\dagger}_{\varphi,k}a_{\varphi,k}},
\end{equation}

{Thus we define a map $J_{\varphi}$} over operators in $\mathcal{F}_{m}$ that is linear, invertible, preserves operator products and conjugation (see the Supplemental Material~\ref{app: JW transformation}), and use it to shift between the {fermionic and anyonic forms of} any operator. In other words, given $O\in\mathcal{F}_{m}$, we define the mapped operators via JWT by the following,
\begin{equation}\label{eq: def_JW}
O=:[O]_{0}=J_{\varphi}([O]_{\varphi}),\quad[O]_{\varphi}=J^{-1}_{\varphi}([O]_{0}).
\end{equation}

Thus, if $[U]_{0}$ is a single-particle change of basis over fermionic states, then $[U]_{\varphi}$ must have the same action in terms of fermionic-anyon states. For fermionic systems, an $N$-particle state $\ket{\phi}$ is \textit{separable}, i.e.,\ it has no particle entanglement, if there is an $N$-particle state in the Fock basis and a single-particle operator $U$ such that \cite{wiseman2003,Plastino_2009}
\begin{equation}\label{eq: N-anyons state}
\ket{\phi}=[U]_{0}(f^{\dagger}_{i_{1}}\cdots f^{\dagger}_{i_{N}})\ket{\textrm{vac}},
\end{equation}
where we can assume that $i_{1}<...<i_{m}$. These states are described by a single Slater determinant, with only exchange correlations due to symmetry \cite{debarba2014,Lourenco2019}. We now extend this notion for general fermionic anyons, leading to the central definition in this work:

\begin{definition}[Separable fermionic anyon states]
A pure $N$-particle fermionic-anyon state $\ket{\phi}$ is separable if and only if there is a {single-particle} fermionic operator $[U]_{0}$ such that
\begin{equation}
\ket{\phi}=[U]_{\varphi}(a^{\dagger}_{\varphi,i_{1}}\cdots {a}^{\dagger}_{\varphi,i_{N}})\ket{\textrm{vac}},
\end{equation} 
{where  $[U]_{\varphi} = J^{-1}_{\varphi}([U]_{0})$.}
\end{definition}

{Having defined a separability criterion, we investigate entanglement in fermionic anyons by adapting corresponding concepts for fermions. For instance, for fermionic systems, Ref.~\cite{gigena2015} shows that single-particle entanglement can be quantified through the minimization over all possible mode representations of
\begin{equation}\label{eq: min_ent}
    E_{SP}(\ket{\psi}) = \min_{f}\sum_i H(\langle f_i^{\dagger}f_i\rangle,\langle f_if_i^{\dagger}\rangle ),
\end{equation}
where $H(p,1-p)=-p\log p - (1-p)\log (1-p)$ is the Shannon binary entropy, and $f_i$ are the fermionic operators transformed according to the Bogoliubov transformation in Eq.~\eqref{eq:canonical}. For fermionic anyons, the entanglement can be obtained by mapping the anyonic state into a fermionic form, calculating the minimization of Eq.~\eqref{eq: min_ent} and translating the new state back into anyonic form. Since $J_{\varphi}$ is a *-algebra endomorphism, we obtain the following theorems (see the  Supplemental Material~\ref{app: single particle entanglement} for proof).
\begin{theorem}[Single-particle entanglement for fermionic anyons]
   For a fermionic-anyon state with fixed particle number, there exists a mode representation such that its single-particle reduced state has the same eigenvalues as the corresponding fermionic single-particle state, which minimizes Eq.~\eqref{eq: min_ent}.
\end{theorem}}

{Theorem 1 implies that there exists a fermionic-anyon mode representation that reflects the particle separability, even though the von Neumann entropy of the single-particle {reduced state}, obtained through the partial trace on another basis, does not characterize the separability. Such a representation ensures that the reduced state is diagonal and independent of the statistical parameter $\varphi$. Consequently, the entropy of the reduced state accurately reflects the entanglement of a single particle to the $N-1$ fermionic anyons. It is possible to generalize the single particle entanglement for mixed states $\rho = \sum_x p_x \ketbra{\psi_x}{\psi_x}$
\begin{equation}
    E(\rho) = \inf_{\{p_x, \ket{\psi_x}\}} \sum_x p_x E_{SP}(\ket{\psi_x})
\end{equation}
 where the $\ket{\psi_x}$ are given in Eq.~\eqref{eq: N-anyons state}. Therefore, for a one-dimensional system with $N$ fermionic anyons, the entanglement between two particles at modes $i$ and $j$ can be computed by taking the partial trace concerning the $N-2$ modes in the minimal entropic basis. This basis is obtained by applying the JWT to the fermionic space. However, when the two particles are in a pure state, it is possible to derive the analogue of a \textit{Schmidt decomposition} for fermionic anyons (as proved in Suplemental Material~\ref{app: slater decomposition}). This involves mapping the state using JWT and calculating the well-known Schliemann decomposition in the fermionic space \cite{schliemann_quantum_2001}}.  
\begin{theorem}[Schmidt decomposition for fermionic anyons]\label{theorem: SD for anyons}
Any pure state of two fermionic anyons with a fixed number of modes has a Schmidt decomposition with the same expansion coefficients as its Schliemann fermionic state counterpart.
\end{theorem}
The decomposition is obtained by a {dressed}  unitary transformation $[U_{SD}]_{{\varphi}}$ that maps a fermionic anyon state $\ket{\psi} = \sum_{m<n}w_{m,n}{a}_{m}^{\dagger} a_n^{\dagger}\ket{0}$ , written in a given basis, onto its Schmidt decomposition $\ket{\psi\pr}$ as 
\begin{equation}\label{eq. dressed schimidt decomposition operation}
   \ket{\psi\pr}=[U_{SD}]_{{\varphi}} \ket{\psi} =  \sum_{\mu}{\omega}_{\mu}\alpha^{\text{\tiny{(1)}}\dagger}_{2\mu}\alpha^{\text{\tiny{(2)}}\dagger}_{2\mu-1}\ket{\textrm{vac}},
\end{equation}
{where $[U_{SD}]_{\varphi}=J^{-1}_{\varphi}[U_{SD}]_{0}$, and $[U_{SD}]_{0}$ is a single-particle fermionic unitary operator that maps the fermionic state $J_{\varphi}(\ket{\psi})$, written in a given basis, in its Schliemann decomposition with coefficients given by $\omega_\mu$.} 

\emph{Fermionic linear optics and fermionic anyons.} To showcase what insights can be drawn from an entanglement theory for fermionic anyons, we apply the formal framework we proposed to particle-based quantum computing. {Specifically, we define a family of computational models based on two-mode ``linear-optical elements'', which reduce to well-known fermionic linear optics when $\varphi=0$, and show how our notion of separability closely tracks the regime of classical simulability of these models.}

Let $PS_{i}$, $BS_{i,j}$ and $PA_{i,j}$ be of the form 
\begin{align*}
PS_{i}(\theta)&=\exp[i\theta(a^{\dagger}_{\varphi,i}a_{\varphi,i})]\\
BS_{i,j}(\theta)&=\exp[i\theta(a^{\dagger}_{\varphi,i}a_{\varphi,j}+a^{\dagger}_{\varphi,j}a_{\varphi,i})]\\
PA_{i,j}(\theta)&=\exp[i\theta(a^{\dagger}_{\varphi,i}a^{\dagger}_{\varphi,j}+a_{\varphi,j}a_{\varphi,i})].
\end{align*}
We refer to these unitaries as \textit{Gaussian optical elements} or, by analogy with linear optics, phase shifters ($PS_{i}$), beam splitters ($BS_{i,j}$), and parametric amplifiers ($PA_{i,j}$). A product of Gaussian optical elements is called an optical circuit. When $\varphi=0$, this set of transformations acting on Fock states and followed by single-mode number detectors defines a computational model called \textit{fermionic linear optics} (FLO). When $\varphi=\pi$, they are called \textit{matchgates} \cite{jozsa_matchgates_2008}, which we refer to here as \textit{qubit linear optics} (QLO). For any other value of $\varphi$, we refer to quantum computing with optical circuits by \textit{fermionic-anyon linear optics} ($\Phi$LO). What operations are analogous to matchgates for fermionic anyons?

Let us use $J_{\varphi}$ to translate known FLO results into results for $\Phi$LO and QLO. First, we look at how fermionic optical elements transform under $J^{-1}_{\varphi}$. We are interested in invariant operations under $J^{-1}_{\varphi}$, i.e., that have the same operator decomposition in all particle systems. Since phase shifters are generated by Hamiltonians proportional to $f^{\dagger}_{i}f_{i}$, they must be invariant under the action of $J^{-1}_{\varphi}$---as must, in fact, be any operator whose fermionic form contains only products of number operators. 

Fermionic beam splitters are generated by Hamiltonians proportional to $f^{\dagger}_{i}f_{j}+f^{\dagger}_{j}f_{i}$. Those are transformed by $J^{-1}_{\varphi}$ into $a^{\dagger}_{\varphi,i}e^{i\varphi\sum^{j-1}_{k=i+1}n_{k}}a_{\varphi,j}+a^{\dagger}_{\varphi,j}e^{-i\varphi\sum^{j-1}_{k=i+1}n_{k}}a_{\varphi,i}.$
This form implies that $J^{-1}_{\varphi}$ leaves only nearest-neighbour beam-splitters invariant. It is known that all single-particle fermionic operators can be decomposed as products of phase shifters and nearest-neighbour beam splitters by using the $f\textsc{SWAP}$ gate, given by
\begin{equation}
f\textsc{SWAP}_{i,i+1}=\exp\left[i\frac{\pi}{2}(f^{\dagger}_{i}-f^{\dagger}_{i+1})(f_{i}-f_{i+1})\right],
\end{equation}
which is itself expressible as a product of nearest-neighbour beam splitters and phase shifters \cite{jozsa_matchgate_2010}. Therefore we conclude that, even if fermionic-anyon {single-particle operators} are complicated, they can always be decomposed in fermionic-anyon nearest-neighbour optical elements.

Fermionic parametric amplifiers are generated by Hamiltonians proportional to $f^{\dagger}_{i}f^{\dagger}_{j}+f_{j}f_{i}$. Their $J^{-1}_{\varphi}$ transforms are given by $e^{-i2\varphi\sum^{i-1}_{k=1}n_{k}}a^{\dagger}_{i}e^{-i\varphi\sum^{j-1}_{k=i+1}n_{k}}a^{\dagger}_{j}+e^{i2\varphi\sum^{i-1}_{k=1}n_{k}}a_{i}e^{i\varphi\sum^{j-1}_{k=i+1}n_{k}}a_{j}.$
The only case where a parametric amplifier is invariant under $J^{-1}_{\varphi}$ is when $\{i,j\}=\{1,2\}$. Nevertheless, we also show in the Supplemental Material \ref{app: fermionic anyons LO} that an arbitrary fermionic $PA_{i,j}$ can be decomposed in terms of $PA_{1,2}$ and $f\textsc{SWAP}$, implying a similar decomposition for their anyonic counterparts.

In \cite{divincenzo_fermionic_2005}, it was shown that FLO circuits are easy to simulate classically in the sense that if $[U]_{0}$ is an FLO circuit, there is a polynomial-time classical algorithm that computes the matrix elements of $[U]_{0}$ in the Fock basis. Now, since the Fock-basis elements of $[U]_{\varphi}$ are, by construction, the same as those of $[U]_{0}$, the same algorithm can efficiently compute the matrix elements of $[U]_{\varphi}$ for states in the fermionic-anyon Fock space. Therefore, any $\Phi$LO circuit composed only of $PA_{1,2}$ and nearest-neighbour beam splitters must be easy to simulate classically in the same sense. For the special case of QLO, this recovers {well-known} simulability results for circuits of nearest-neighbour matchgates \cite{knill_fermionic_2001}.

Given that all FLO circuits are easy to simulate classically represented either in fermionic or anyonic form, we might ask if all $\Phi$LO circuits are also easy to simulate (mode details are presented in Supplemental Material \ref{app: classically simulable}). The answer, however, is no, for the following reasons. A fermionic-anyon beam splitter is generated by a Hamiltonian proportional to $a_{\varphi,i}^{\dagger}a_{\varphi,j}+a_{\varphi,j}^{\dagger}a_{\varphi,i}$. Under $J_{\varphi}$, this gets transformed into $f^{\dagger}_{i}e^{-i\varphi\sum^{j-1}_{k=i+1}n_{k}}f_{j}+f^{\dagger}_{j}e^{i\varphi\sum^{j-1}_{k=i+1}n_{k}}f_{i},$
which only generates a fermionic Bogoliubov transformation if $j=i+1$ and, therefore, is not generally an FLO circuit. It was shown in \cite{tosta_quantum_2019} that non-nearest neighbour beam splitters allow for universal quantum computation with fermionic anyons for all $\varphi \neq 0$, {which also reduces to a known result for matchgates when $\varphi = \pi$} \cite{jozsa_matchgates_2008}. {Furthermore, it is known that almost all fermionic non-Gaussian operations (i.e., gates outside of FLO) can extend it to universality \cite{michal17}, from which follows the analogous statement for $\Phi$LO for any $\varphi \neq 0$.}

To summarize, the set of FLO circuits is strictly smaller than the set of $\Phi$LO (or QLO) circuits. Furthermore, the $J^{-1}_{\varphi}$ map sends FLO circuits into a small subset of $\Phi$LO circuits, which is particularly easy to simulate classically (when acting on the Fock basis). What about the computational power of these models with input states not on the Fock basis?

In \cite{hebenstreit_all_2019}, the authors show a magic-state injection protocol that uses only nearest-neighbour QLO operations to perform universal quantum computation. Furthermore, they also show that \emph{any} fermionic non-Gaussian state is a magic state for the same protocol. Since the transformations themselves are non-universal, we can identify a computational \emph{resource}, necessary for a quantum speedup, in the magic states and identify the set of fermionic Gaussian states as resource-free. This dichotomy matches that defined by the notion of separability: (pure) Gaussian fermionic states are also the free states if one views entanglement as the resource, as done in \cite{gigena_one-body_2020} based on a definition of one-body entanglement entropy.

{The proposed methods allow us to repurpose these previous results and draw similar conclusions for fermionic anyons. By writing the magic state injection protocol in terms of $J_{\varphi}^{-1}$ invariant optical elements, and subsequently applying $J^{-1}_{\varphi}\circ J_{\pi}$ to the corresponding circuit, the same injection protocol can use fermionic-anyon magic states to induce a non-FLO operation. Following the approach developed in Ref.\cite{gigena_one-body_2020}, our results imply that the notions of free states for both types of resources (computational power and entanglement) match for fermionic anyons as they do for fermions.}

{Our formalism can also be used to understand previous results about {matchgate circuits (i.e., QLO)}. References \cite{jozsa_matchgates_2008, brod_extending_2011,brod_geometries_2012}, for example, consider supplementing circuits of nearest-neighbour matchgates with other resources. The authors use a dual-rail encoding, where we can encode a logical 0 (resp.\ 1) qubit state as the $\ket{01}$ (resp.\ $\ket{10}$) state of two physical qubits. In that case, ref.\  \cite{jozsa_matchgate_2010} shows that the $f\textsc{swap}$ gate cannot be used to generate entanglement between the two logical qubits, whereas the \textsc{swap} gate can---a curious role reversal, given that the $f\textsc{swap}$ is a maximally entangling two-qubit gate and the \textsc{swap} is not entangling. Our formalism {provides an alternative interpretation that} resolves this conundrum neatly: there is a notion of entanglement between logical qubits in a matchgate circuit, {corresponding to the one we proposed}, when one views the state of a physical qubit as the occupation number of a fermionic-anyon mode (at $\varphi=\pi$). {This notion of entanglement would naturally differ from the standard definition of entanglement between the physical qubits, but would be more relevant to the computational complexity of matchgate circuits with non-Gaussian elements - for instance, it is an interesting question for future work whether this alternative notion of entanglement translates into a quantitative measure of the complexity of classical simulation of matchgate circuits.}}

%\clearpage
\emph{Conclusions.}
In summary, we have introduced a resource-theoretic framework for investigating the separability of fermionic anyons and their connection to quantum computing. We characterized the entanglement of   fermionic anyons, and showed that the concept of fermionic-anyon separability can be mapped to the free resources of matchgate circuits. Our framework was applied to particle-based quantum computing, revealing that fermionic-anyon linear-optical circuits can be expressed using nearest-neighbour beam splitters, phase shifters, and mode swaps. Additionally, we showed that universal quantum computation with fermionic anyons can be achieved by introducing non-separable states, similar to the magic-state injection protocol presented in Ref.~\cite{hebenstreit_all_2019}. 

Finally, we translated the statistical parameter $\varphi$ of the fermionic-anyon commutation relation into two well-established forms of universal quantum computation: fermionic and qubit-based. When $\varphi=0$, universal anyonic quantum computation reduces to fermionic linear optics; similarly, qubit linear optics can be obtained by interpreting a physical qubit as the occupation of a fermionic-anyon mode at $\varphi=\pi$. This approach creates a matchgate scheme where magic states are entangled as per our definition rather than the traditional notion for qubits. These notions are not equivalent, and our definition is instead a type of particle entanglement when one interprets qubits as occupation numbers of exotic particles \cite{WuLydar2002}. Nonetheless, we consider it to have already helped to reinterpret the results of \cite{jozsa_matchgates_2008} in a clearer manner. We leave it as a direction for future research to investigate further consequences of viewing qubit circuits via the lens of our definition of fermionic-anyon entanglement, {as well as the possible resourceful limitations and costs of such a quantum computational model. Naturally, one could ask: what quantum computational models exist for intermediate parameters $\varphi\in (0,\pi)$? We leave this question unanswered and also propose a potential generalization model for two-dimensional fermionic anyons, inviting further investigation.}

{To achieve a comprehensive quantum computation framework, we need to establish a measurement method specific to fermionic anyons. One can explore the measurement disturbance model presented in \cite{debarba2017} and adapt it for fermionic anyons employing JWT. Considering the techniques outlined in \cite{gigena_one-body_2020}, one could fully describe the theoretical resources available for fermionic anyons computation.}

\emph{Acknowledgements.}
  The authors acknowledge support from the Brazilian agency CNPq INCT-IQ through the project (465469/2014-0). AT also acknowledges support from the Serrapilheira Institute (grant number Serra-1709-17173). ACL also acknowledges support from FAPESC. DJB also acknowledges support from CNPq and FAPERJ (Grant No.\ $202.782/2019$). FI also acknowledges support from  CNPq (Grant No. $308637/2022-4$), FAPERJ (Grant No. E-$26/211.318/2019$ and E-$26/201.365/2022$) and Alexander von Humboldt Foundation. TD also acknowledges support from CNPq (Grant No. $441774/2023-7$ and $200013/2024-6$), European Research Council (Consolidator grant ‘Cocoquest’ 101043705) and ÖAW-JESH-Programme. 
% %%%%%%%%%%%%%%%%%%%%%%%%%%%%%%%%%%%%%%%%%%%%%%%%%%%%%%%%%%%%%%%
% %%%%%%%%%%%%%%%%%%%%%%%%%%%%%%%%%%%%%%%%%%%%%    
\let\oldaddcontentsline\addcontentsline% Store \addcontentsline
\renewcommand{\addcontentsline}[3]{}% Make \addcontentsline a no-op

\let\addcontentsline\oldaddcontentsline% Restore \addcontentsline

%%%%%%%%%%%%%%%%%%%%%%%%%%%%%%%%%%%%%%%%

%%%%%%%%%%%%%%%%%%%%%%%%%%%%%%%%%%%%%%%%

\newpage
\onecolumngrid
%\appendix
\newpage
%\starttocentries
%\vspace*{-5mm}
\pdfbookmark[0]{Supplemental Material}{Supplemental Material}
\vspace*{5mm}
\label{supplementalmaterial}
\phantomsection
\begin{center}
\begin{LARGE}
Supplemental Material
\end{LARGE}
\end{center}

{In this Supplemental Material we give further details of the analytical results which have been omitted from the main text. Specifically, we provide an expanded and mathematically rigorous discussion of the fermionic-anyon algebra, their Jordan-Wigner transformation and subsequent connections to fermionic linear optics and classical simulability.  We moreover present an illustrative example of such ideas based in a beam splitter circuit with four modes composed of two fermionic anyons. 
}

\vspace*{-5mm}

\renewcommand{\thesection}{\Alph{section}}
\renewcommand{\thesubsection}{\Alph{section}\arabic{subsection}}
\renewcommand{\thesubsubsection}{\Alph{section}\arabic{subsection}\alph{subsubsection}}
\makeatletter
\renewcommand{\p@subsection}{}
\renewcommand{\p@subsubsection}{}
\makeatother

% %%%%%

\tableofcontents

\section{Fermionic-anyon algebras and the Jordan-Wigner map}\label{app: anyon_description}

In this section, we formally describe the operator algebras of fermionic anyons as *-algebras and prove that they admit a particle interpretation. The formal *-algebra structure is what allows us to give a full characterisation of the statistical transmutation maps, e.g., Jordan-Wigner maps of standard and fractional varieties, as *-algebra homomorphisms. We finish by proving that statistical transmutation preserves the amplitudes of multiparticle states in the Fock-basis.

\subsection{Fermionic-anyon algebras as *-algebras}

Here we define what are *-algebras over the complex numbers, show a particular method of constructing them, and prove that the observable algebra of all fermionic anyons are examples of *-algebras under this construction.

\begin{defi}\label{def: *-algebra_def}
Let $V$ be a vector space over $\mathbb{C}$. We call $V$ an \emph{algebra} if it is equipped with bilinear product $(\cdot):V \times V\rightarrow V$ between vectors and a special vector $e$ that satisfies $e.v=v.e=v$ for any $v\in V$. We call $e$ the unit of the algebra. A \textit{*-algebra}, is an algebra with a unitary operation $\dagger:V\rightarrow V$, called the conjugate, that satisfies;
\begin{itemize}
     \item[1.] for all $v,u\in V$, $(v+u)^{\dagger}=v^{\dagger}+u^{\dagger}$,
    \item[2.] for all $v,u\in V$, $(vu)^{\dagger}=u^{\dagger}v^{\dagger}$,
    \item[3.] $e^{\dagger}=e$, and
    \item[4.] for any $v\in V$, $(v^{\dagger})^{\dagger}=v$.
\end{itemize}
\end{defi}

Essentially, the *-algebra is a generalisation of the concept of an algebra of observables, where physical observables are represented by self-conjugated elements. The *-algebras we consider here are those where all elements can be written as linear combinations of products of a small set of ``primitive'' elements. This idea is captured by the definition of a free *-algebra.

\begin{defi} \label{def: free_*_algebra}
The \emph{free *-algebra} $F(S)$ generated by the set of formal symbols $S=\{x_{i}|i\in\Lambda\}$, where $\Lambda$ is some index set, is the *-algebra whose elements are linear combinations of all possible strings over the symbols in $S$ and their formal conjugates, where the product is string concatenation. The elements of $S$ are called the generators of $F(S)$. When $\Lambda$ is a finite set, we say that $F(S)$ is \textit{finitely-generated} and the number of generators is called its rank.
\end{defi}

We can use the free *-algebra to define other *-algebras by imposing relations between its elements. This allows us to characterise the algebras of particle operators as special kinds of *-algebras.

\begin{defi}\label{def: ferm_any_algebra_def}
The \emph{fermionic-anyon algebra} over $m$ modes with statistical parameter $\varphi$, denoted by $\mathcal{A}^{\varphi}_{m}$, is the quotient of the free *-algebra generated by the set $\{a_{\varphi,i}|i\in\{1,...,m\}\}$ by the relations
\begin{align}
&a_{\varphi,i}a_{\varphi,j}^{\dagger}+e^{-i\varphi\epsilon_{ij}}a_{\varphi,j}^{\dagger}a_{\varphi,i}=\delta_{ij} \\
&a_{\varphi,i}a_{\varphi,j}+e^{i\varphi\epsilon_{ij}}a_{\varphi,j}a_{\varphi,i}=0, 
\end{align}
with $\epsilon_{ij}$ given by
\begin{equation*}
    \epsilon_{ij}=\left\{\begin{matrix}
1, & \textnormal{if} \ \ i<j\\ 
0, & \textnormal{if} \ \ i=j\\ 
-1, & \textnormal{if} \ \ i>j.
\end{matrix}\right.
\end{equation*}
\end{defi}

When $\varphi=0$, the above relations reduce to canonical anticommutation relations, which implies that the special case $\mathcal{A}^{0}_{m}$ is the *-algebra of fermionic operators, called $\mathcal{F}_{m}$, and that we can define $a_{0,i}:=f_{i}$. When $\varphi=\pi$, those same relations reduce to the commutation relations between raising and lowering operators for spin-1/2 systems, implying that $\mathcal{A}^{\pi}_{m}$ is the observable algebra of a system of $m$ qubits, which we refer to as $\mathcal{Q}_{m}$. Finally, note that $\mathcal{A}^{\varphi}_{m}$ is the exact same *-algebra as $\mathcal{A}^{\varphi+2\pi n}_{m}$ for all $n\in\mathbb{Z}$.

\subsection{*-Algebras and their particle interpretation}

Up to this point, we referred to the *-algebras $\mathcal{A}^{\varphi}_{m}$ as ``fermionic-anyon algebras'', but we only gave a physical interpretation to the special cases $\varphi=0,\pi$. Now we prove that the $\mathcal{A}^{\varphi}_{m}$ have a physical interpretation as observable algebras of particle systems, and are not just abstract constructions. This interpretation comes from the existence of a Fock-state representation.

\begin{defi}\label{def: exist_Fock_rep}
Let $V$ be a *-algebra of rank $m$ over $\mathbb{C}$. $V$ is said to possess a Fock representation if there is a set of  $m$ self-conjugate elements $n_{i}\in V$ and a set of $m$ generators $l_{i}\in V$ such that for all $i,j=1,...,m$
\begin{equation}
[n_{i},n_{j}]=0\text{, and }[n_{i},l_{j}]=-\delta_{i,j}l_{j}
\end{equation}
\end{defi}

If a *-algebra satisfies this definition, we can call the $n_{i}$ operators its \emph{number operators}, and the generators $l_{i}$ its \emph{annihilation operators}, while the *-algebra axioms guarantee that the $l^{\dagger}_{i}$ elements exist and behave as creation operators. This allows us to construct the associated Fock space.

\begin{defi}\label{def: cons_Fock_rep}
Let $V$ be a *-algebra of rank $m$ that possesses a Fock representation. The Fock space for $V$ is the vector space $\mathcal{H}_{V}$ with basis set given by all non-zero vectors of the form
\begin{equation}
l^{\dagger \alpha_{1}}_{1}\cdots l^{\dagger \alpha_{m}}_{m}\ket{0}_{V},
\end{equation}
where each $\alpha_{i}$ is a natural number between $0$ and $m_{i}$, $m_{i}$ is the highest natural number such that $l^{m}_{i}\neq0$, and $\ket{0}_{V}$ is the only state in $\mathcal{H}_{V}$ satisfying
\begin{equation}
l_{i}\ket{0}_{V}=0,
\end{equation}
for all $i=1,...,m$.
\end{defi}

With the definition and method for constructing a Fock representation, we now prove the main claim of this section.

\begin{prop}\label{prop: ferm_any_Fock_rep}
The algebra $\mathcal{A}^{\varphi}_{m}$ has a Fock space representation.
\end{prop}

\begin{proof}
Consider the operator $n_{\varphi,i}=a^{\dagger}_{\varphi,i}a_{\varphi,i}$. From the commutation relations for fermionic anyons, we can see that
\begin{align*}
[n_{\varphi,i},a_{\varphi,j}]&=a^{\dagger}_{\varphi,i}a_{\varphi,i}a_{\varphi,j}-a_{\varphi,j}a^{\dagger}_{\varphi,i}a_{\varphi,i}\\&=-e^{i\varphi\epsilon_{i,j}}a^{\dagger}_{\varphi,i}a_{\varphi,j}a_{\varphi,i}-a_{\varphi,j}a^{\dagger}_{\varphi,i}a_{\varphi,i}\\
&=-\delta_{i,j}a_{\varphi,i}+a_{\varphi,j}a^{\dagger}_{\varphi,i}a_{\varphi,i}-a_{\varphi,j}a^{\dagger}_{\varphi,i}a_{\varphi,i}\\&=-\delta_{i,j}a_{\varphi,j}.
\end{align*}
Given the form of $n_{\varphi,i}$, the last equation implies that $n_{\varphi,i}n_{\varphi,j}=n_{\varphi,j}n_{\varphi,i}$. Therefore, the Fock-space $\mathcal{H}_{A^{\varphi}_{m}}$, which we call $\mathcal{H}_{\varphi}$ exists, and has a basis
\begin{equation}
a^{\dagger \alpha_{1}}_{\varphi,1}\cdots a^{\dagger \alpha_{m}}_{\varphi,m}\ket{0}_{\varphi},
\end{equation}
where $\ket{0}_{\varphi}$ is the vacuum state, and in this case $\alpha_{i}\in\{0,1\}$.
\end{proof}

\section{Jordan-Wigner transformation and *-algebra homomorphisms}\label{app: JW transformation}

Having a solid grasp on the properties of fermionic-anyon algebras as *-algebras, we now discuss what *-algebra homomorphisms are, how they are constructed, and why the Jordan-Wigner transform and their generalisations should be viewed as *-algebra homomorphisms.

\begin{defi}\label{def: *-alg_homo}
Let $V$ and $U$ be two *-algebras. A map $g:V\rightarrow U$ is a \textit{*-algebra homomorphism} if and only if
\begin{itemize}
    \item[1.] for all $v,u\in V$, $g(u+v)=g(u)+g(v)$,
    \item[2.] for all $v\in V$ and $c\in\mathbb{C}$, $g(cv)=cg(v)$,
    \item[3.] for all $v,u\in V$, $g(uv)=g(u)g(v)$,
    \item[4.] for all $v\in V$, $g(v^{\dagger})=g(v)^{\dagger}$,
    \item[5.] and $g(e)=e$,
\end{itemize}
where, for each equation, the operations on the left take place in $V$, and the operations on the right take place in $U$.
\end{defi}

As such, a *-algebra homomorphism is any map that preserves the *-algebra's conjugation, multiplication and addition. If our *-algebras are finitely generated, we can actually describe any *-algebra homomorphisms by specifying their action over its generators.

\begin{prop}\label{prop: img_gen_homo}
Let $V,U$ be *-algebras of rank $m$, and let $g:V\rightarrow U$ be an algebra homomorphism. Then, the image of any $\alpha\in V$ is a function of the images of $\{l_{i}|i=,1...,m\}\in V$, where $l_{i}$ are generators of $V$.
\end{prop}

\begin{proof}
If $\{l_{i}|i=,1...,m\}$ is a set of generators of $V$, then any element $\alpha\in V$ has the form
\begin{equation}
\alpha=\sum_{i_{1},...,i_{2m}\in I}c_{i_{1},...,i_{2m}}l^{\dagger i_{1}}_{1}\cdots{l}^{\dagger i_{m}}_{m}l^{i_{m+1}}_{1}\cdots{l}^{i_{2m}}_{m},
\end{equation}
where $c_{i_{1},...,i_{2m}}\in\mathbb{C}$, and where $I=\{(0,...,0),...,(\mu_{1},...,\mu_{2m})\}$ with $\mu_{i}$ being the largest natural number such that $l^{\mu_{i}}_{i}\neq0$. Then, by the definition of *-algebra homomorphisms, we have that
\begin{align*}
g(\alpha)&=\sum_{i_{1},...,i_{2m}\in I}g(c_{i_{1},...,i_{2m}}l^{\dagger i_{1}}_{1}\cdots{l}^{\dagger i_{m}}_{m}l^{i_{m+1}}_{1}\cdots{l}^{i_{2m}}_{m})\\
&=\sum_{i_{1},...,i_{2m}\in I}c_{i_{1},...,i_{2m}}g(l^{\dagger i_{1}}_{1}\cdots{l}^{\dagger i_{m}}_{m}l^{i_{m+1}}_{1}\cdots{l}^{i_{2m}}_{m})=\\
&=\sum_{i_{1},...,i_{2m}\in I}c_{i_{1},...,i_{2m}}g(l^{\dagger i_{1}}_{1})\cdots g({l}^{\dagger i_{m}}_{m})g(l^{i_{m+1}}_{1})\cdots g({l}^{i_{2m}}_{m})\\
&=\sum_{i_{1},...,i_{2m}\in I}c_{i_{1},...,i_{2m}}g(l_{1})^{\dagger i_{1}}\cdots g(l_{m})^{\dagger i_{m}}g(l_{1})^{i_{m+1}}\cdots g(l_{m})^{i_{2m}},
\end{align*}
as we wished to show.
\end{proof}

In fact, we can also show that *-algebra homomorphisms preserve all algebraic equations relating elements inside a *-algebra.

\begin{coro}\label{cor: val_alg_rel}
Let $V$ be a *-algebra and let $\alpha_{1},...,\alpha_{\nu}\in V$ be a finite-sized set of elements of $V$, such that
\begin{equation}
f(\alpha_{1},...,\alpha_{\nu})=0
\end{equation}
with $f$ being any formal power series. Then, if $U$ is a *-algebra and $g:V\rightarrow U$ a *-algebra homomorphism we have that
\begin{equation}
f(g(\alpha_{1}),...,g(\alpha_{\nu}))=0.
\end{equation}
\end{coro}

\begin{proof}
We use a proof strategy similar to the one used in the previous proposition. Given that a formal power series involves only multiplication by complex numbers, and the *-algebra addition, product and conjugation operations, we will always be able to ``pass through'' these operations with the action of $g$ until we reach the $\alpha_{i}$ of $V$, proving the statement.
\end{proof}

The fact that *-algebra homomorphisms can be used to convert valid algebraic relations from one *-algebra to another is the main reason why we choose this machinery to describe the Jordan-Wigner transform and its generalisations, as we now explain. If we consider fermionic-anyon creation and annihilation operators as independently defined objects for every value of $\varphi$, then a Jordan-Wigner transformation should be able to map the creation and annihilation operators for a particular value $\varphi_{1}$ as functions of their counterparts of another value $\varphi_{2}$ in such a way that this expression in $\varphi_{2}$ obeys the commutation relations of creation and annihilation operators in $\varphi_{1}$. This is exactly what is accomplished by a *-algebra, and is what inspires the following definition.

\begin{defi}\label{def: exch_trans_map}
For any $\varphi_{1},\varphi_{2}\in[0,2\pi)$, let $J_{\varphi_{1},\varphi_{2}}:\mathcal{A}^{\varphi_{1}}_{m}\rightarrow\mathcal{A}^{\varphi_{2}}_{m}$ be the *-algebra homomorphism defined by
\begin{equation}
J_{\varphi_{1},\varphi_{2}}(a_{\varphi_{1},i})=a_{\varphi_{2},i}e^{i(\varphi_{2}-\varphi_{1})\sum_{k=1}^{i-1}n_{\varphi_{2},k}}.
\end{equation}
We call these *-algebra homomorphisms \textbf{exchange transmutation maps}. The $J_{\varphi,0}$ maps are called \textbf{fractional Jordan-Wigner transforms} and are denoted by $J_{\varphi}$, while the special case $J_{\pi}$ is called the \textbf{standard Jordan-Wigner transform}.
\end{defi}

We now prove that this *-algebra homomorphism is well-defined

\begin{prop}\label{prop: jor_wig_prop}
For any $\varphi_{1},\varphi_{2}\in[0,2\pi)$, we have that
\begin{align}
&J_{\varphi_{1},\varphi_{2}}(a_{\varphi_{1},i})J^{\dagger}_{\varphi_{1},\varphi_{2}}(a_{\varphi_{1},j})+e^{-i\varphi_{1}\epsilon_{ij}}J^{\dagger}_{\varphi_{1},\varphi_{2}}(a_{\varphi_{1},j})J_{\varphi_{1},\varphi_{2}}(a_{\varphi_{1},i})=\delta_{ij} \\
&J_{\varphi_{1},\varphi_{2}}(a_{\varphi_{1},i})J_{\varphi_{1},\varphi_{2}}(a_{\varphi_{1},j})+e^{i\varphi\epsilon_{ij}}J_{\varphi_{1},\varphi_{2}}(a_{\varphi_{1},j})J_{\varphi_{1},\varphi_{2}}(a_{\varphi_{1},i})=0, 
\end{align}
\end{prop}

\begin{proof}
This corollary is equivalent to affirming that the $\mathcal{A}^{\varphi_{2}}_{m}$ elements
\begin{equation}
a_{\varphi_{2},i}e^{i(\varphi_{2}-\varphi_{1})\sum_{k=1}^{i-1}n_{\varphi_{2},k}},
\end{equation}
obey the same commutation relations as $a_{\varphi_{1},i}$ for all $i$, which are generators of $\mathcal{A}^{\varphi_{1}}_{m}$. In the $i=j$ case we have that
\begin{align*}
\begin{split}
&\left(a_{\varphi_{2},i}e^{i(\varphi_{2}-\varphi_{1})\sum_{k=1}^{i-1}n_{\varphi_{2},k}}\right)\left(a^{\dagger}_{\varphi_{2},i}e^{-i(\varphi_{2}-\varphi_{1})\sum_{k=1}^{i-1}n_{\varphi_{2},k}}\right)+\left(a^{\dagger}_{\varphi_{2},i}e^{-i(\varphi_{2}-\varphi_{1})\sum_{k=1}^{i-1}n_{\varphi_{2},k}}\right)\left(a_{\varphi_{2},i}e^{i(\varphi_{2}-\varphi_{1})\sum_{k=1}^{i-1}n_{\varphi_{2},k}}\right)\\
=&\left(a_{\varphi_{2},i}e^{i(\varphi_{2}-\varphi_{1})n_{\varphi_{2},i}}\right)\left(a^{\dagger}_{\varphi_{2},i}e^{-i(\varphi_{2}-\varphi_{1})n_{\varphi_{2},i}}\right)+\left(a^{\dagger}_{\varphi_{2},i}e^{-i(\varphi_{2}-\varphi_{1})n_{\varphi_{2},i}}\right)\left(a_{\varphi_{2},i}e^{i(\varphi_{2}-\varphi_{1})n_{\varphi_{2},i}}\right)\\
=&\; e^{i(\varphi_{1}-\varphi_{2})}\left(a_{\varphi_{2},i}a^{\dagger}_{\varphi_{2},i}+a^{\dagger}_{\varphi_{2},i}a_{\varphi_{2},i}\right)\\
=&\;1
\end{split}
\end{align*}
{and along the same line of reasoning,}
\begin{align*}
&\left(a_{\varphi_{2},i}e^{i(\varphi_{2}-\varphi_{1})\sum_{k=1}^{i-1}n_{\varphi_{2},k}}\right)\left(a_{\varphi_{2},i}e^{i(\varphi_{2}-\varphi_{1})\sum_{k=1}^{i-1}n_{\varphi_{2},k}}\right)+\left(a_{\varphi_{2},i}e^{i(\varphi_{2}-\varphi_{1})\sum_{k=1}^{i-1}n_{\varphi_{2},k}}\right)\left(a_{\varphi_{2},i}e^{i(\varphi_{2}-\varphi_{1})\sum_{k=1}^{i-1}n_{\varphi_{2},k}}\right)=\\
=&\;2a_{\varphi_{2},i}\left(e^{i(\varphi_{2}-\varphi_{1})\sum_{k=1}^{i-1}n_{\varphi_{2},k}}\right)a_{\varphi_{2},i}\left(e^{i(\varphi_{2}-\varphi_{1})\sum_{k=1}^{i-1}n_{\varphi_{2},k}}\right)=\\
=&\;2(a_{\varphi_{2},i}^{2})\left(e^{2i(\varphi_{2}-\varphi_{1})\sum_{k=1}^{i-1}n_{\varphi_{2},k}}\right)\\
=&\;0,
\end{align*}
{On the other hand, for the $i<j$ case we see that}
\begin{align*}
&\left(a_{\varphi_{2},i}e^{i(\varphi_{2}-\varphi_{1})\sum_{k=1}^{i-1}n_{\varphi_{2},k}}\right)\left(a^{\dagger}_{\varphi_{2},j}e^{-i(\varphi_{2}-\varphi_{1})\sum_{k=1}^{j-1}n_{\varphi_{2},k}}\right)+e^{-i\varphi_{1}}\left(a^{\dagger}_{\varphi_{2},j}e^{-i(\varphi_{2}-\varphi_{1})\sum_{k=1}^{j-1}n_{\varphi_{2},k}}\right)\left(a_{\varphi_{2},i}e^{i(\varphi_{2}-\varphi_{1})\sum_{k=1}^{i-1}n_{\varphi_{2},k}}\right)\\
=& \; a_{\varphi_{2},i}a^{\dagger}_{\varphi_{2},j}\left(e^{-i(\varphi_{2}-\varphi_{1})\sum_{k=i-1}^{j-1}n_{\varphi_{2},k}}\right)+e^{-i\varphi_{1}}e^{-i(\varphi_{2}-\varphi_{1})}a^{\dagger}_{\varphi_{2},j}a_{\varphi_{2},i}\left(e^{-i(\varphi_{2}-\varphi_{1})\sum_{k=i-1}^{j-1}n_{\varphi_{2},k}}\right)\\
=&\left(a_{\varphi_{2},i}a^{\dagger}_{\varphi_{2},j}+e^{-i\varphi_{2}}a^{\dagger}_{\varphi_{2},j}a_{\varphi_{2},i}\right)e^{-i(\varphi_{2}-\varphi_{1})\sum_{k=i-1}^{j-1}n_{\varphi_{2},k}}\\
=& \; 0,
\end{align*}
{and similarly,}
\begin{align*} 
&\left(a_{\varphi_{2},i}e^{i(\varphi_{2}-\varphi_{1})\sum_{k=1}^{i-1}n_{\varphi_{2},k}}\right)\left(a_{\varphi_{2},j}e^{i(\varphi_{2}-\varphi_{1})\sum_{k=1}^{j-1}n_{\varphi_{2},k}}\right)+e^{i\varphi_{1}}\left(a_{\varphi_{2},j}e^{i(\varphi_{2}-\varphi_{1})\sum_{k=1}^{j-1}n_{\varphi_{2},k}}\right)\left(a_{\varphi_{2},i}e^{i(\varphi_{2}-\varphi_{1})\sum_{k=1}^{i-1}n_{\varphi_{2},k}}\right)=\\
=&\;a_{\varphi_{2},i}a_{\varphi_{2},j}\left(e^{i(\varphi_{2}-\varphi_{1})\left(\sum_{k=1}^{j-1}n_{\varphi_{2},k}+\sum_{l=1}^{i-1}n_{\varphi_{2},l}\right)}\right)+e^{i\varphi_{1}}e^{i(\varphi_{2}-\varphi_{1})}a_{\varphi_{2},j}a_{\varphi_{2},i}\left(e^{i(\varphi_{2}-\varphi_{1})\left(\sum_{k=1}^{j-1}n_{\varphi_{2},k}+\sum_{l=1}^{i-1}n_{\varphi_{2},l}\right)}\right)=\\
=&\left(a_{\varphi_{2},i}a_{\varphi_{2},j}+e^{i\varphi_{2}}a_{\varphi_{2},j}a_{\varphi_{2},i}\right)\left(e^{i(\varphi_{2}-\varphi_{1})\left(\sum_{k=1}^{j-1}n_{\varphi_{2},k}+\sum_{l=1}^{i-1}n_{\varphi_{2},l}\right)}\right)\\
=&\;0.
\end{align*}
Given that the proof for the $i>j$ case is obtained taking the conjugate of the $i<j$ case, we have successfully proven that $J_{\varphi_{1},\varphi_{2}}$ is well-defined for all $\varphi_{1},\varphi_{2}\in[0,2\pi)$.
\end{proof}

Therefore, these maps capture the desired behaviour of a Jordan-Wigner transform. Now, we can investigate their properties.

\begin{prop}\label{prop: group_struct}
For all $i\in\{1,...,m\}$ it holds that 
\begin{equation}
J_{\varphi_{1},\varphi_{2}}(n_{\varphi_{1},i})=n_{\varphi_{2},i},
\end{equation}
and
\begin{equation}\label{eq: group_clos}
J_{\varphi_{2},\varphi_{3}}\circ J_{\varphi_{1},\varphi_{2}}=J_{\varphi_{1},\varphi_{3}},
\end{equation}
for any $\varphi_{1},\varphi_{2},\varphi_{3} \in [0,2\pi)$.
\end{prop}

\begin{proof}
To prove the first part, notice that $J_{\varphi_{1},\varphi_{2}}(a^{\dagger}_{\varphi_{1},i})=(J_{\varphi_{1},\varphi_{2}}(a_{\varphi_{1},i}))^{\dagger}$, from which it follows that
\begin{align*}
J_{\varphi_{1},\varphi_{2}}(a^{\dagger}_{\varphi_{1},i}a_{\varphi_{1},i})&=J_{\varphi_{1},\varphi_{2}}(a^{\dagger}_{\varphi_{1},i})J_{\varphi_{1},\varphi_{2}}(a_{\varphi_{1},i})\\&=a^{\dagger}_{\varphi_{2},i}e^{-i(\varphi_{2}-\varphi_{1})\sum_{k=1}^{i-1}n_{\varphi_{2},k}}a_{\varphi_{2},i}e^{i(\varphi_{2}-\varphi_{1})\sum_{k=1}^{i-1}n_{\varphi_{2},k}}\\&=a^{\dagger}_{\varphi_{2},i}a_{\varphi_{2},i}.
\end{align*}
Subsequently,
\begin{align*}
(J_{\varphi_{2},\varphi_{3}}\circ J_{\varphi_{1},\varphi_{2}})(a_{\varphi_{1},i})
&=J_{\varphi_{2},\varphi_{3}}\left(a_{\varphi_{2},i}e^{i(\varphi_{2}-\varphi_{1})\sum_{k=1}^{i-1}n_{\varphi_{2},k}}\right)\\
&=J_{\varphi_{2},\varphi_{3}}(a_{\varphi_{2},i})e^{i(\varphi_{2}-\varphi_{1})\sum_{k=1}^{i-1}J_{\varphi_{2},\varphi_{3}}(n_{\varphi_{2},k})}\\
&=a_{\varphi_{3},i}e^{i(\varphi_{2}-\varphi_{1})\sum_{k=1}^{i-1}n_{\varphi_{3},k}}e^{i(\varphi_{3}-\varphi_{2})\sum_{k=1}^{i-1}n_{\varphi_{3},k}}\\
&=a_{\varphi_{3},i}e^{i(\varphi_{3}-\varphi_{1})\sum_{k=1}^{i-1}n_{\varphi_{3},k}}\\
&=J_{\varphi_{1},\varphi_{3}}(a_{\varphi_{1},i}).
\end{align*}
\end{proof}

\begin{coro}\label{coro: path_prop}
For any $\varphi\in[0,2\pi)$ and $k,l\in\mathbb{Z}$, we have that $J_{\varphi+2\pi k,\varphi+ 2\pi l}=\mathds{1}_{\varphi}$, where $\mathds{1}_{\varphi}$ is the identity homomorphism over $\mathcal{A}^{\varphi}_{m}$. This implies that, for any $\varphi_{1},\varphi_{2}\in[0,2\pi)$,
\begin{equation}
J^{-1}_{\varphi_{1},\varphi_{2}}=J_{\varphi_{2},\varphi_{1}}.
\end{equation}
\end{coro}

\begin{proof}
This is a direct consequence of Eq.~(\ref{eq: group_clos}).
\end{proof}

The main takeaway is that these properties show that all $\mathcal{A}^{\varphi}_{m}$ are isomorphic to each other as *-algebras. Therefore, exchange transmutation maps not only transfer true algebraic relations from one algebra to another, but they also act as a kind of ``generator basis transformation'' of the observable algebra of these particles, implying that the algebras $\mathcal{A}^{\varphi}_{m}$ are just alternate representations of the algebra of fermionic observables. In the next section, we investigate the consequences of this fact for the Fock representation of these particle systems.

\subsection{Invariance of state amplitudes}

Following up on the idea of fermionic anyons being alternate representations of the fermionic observable algebra, we now restrict our discussion to maps between fermions and fermionic anyons, and talk of ``operators'' in an abstract way, specifying their expression in creation and annihilation operators of a particular fermionic-anyon algebra as a ``basis choice''. Using this language, we can prove that fermionic Fock states have a ``basis-invariant'' description. This implies that amplitudes of fermionic states do not change under the change of exchange statistics, which is crucial for our concept of state separability discussed in the main text.

\begin{defi}
For any $\varphi\in[0,2\pi)$, let $O$ be an operator in $\mathcal{F}_{m}$. For all $\varphi\in[0,2\pi)$, the \emph{${\varphi}$-representation of O} is the operator defined by the following:
\begin{equation}
O=:[O]_{0}=J_{\varphi}([O]_{\varphi}),\quad[O]_{\varphi}=J^{-1}_{\varphi}([O]_{0}).
\end{equation}
\end{defi}

The last definition introduces the notation for describing the fractional Jordan-Wigner transform as a basis change in operator space. This allows us to completely specify any operator on any fermionic-anyon algebra by their expansion coefficients in the fermionic algebra. We now apply this to operators that create Fock-basis states.

\begin{prop}
Let $\ket{\mathbf{x}}_{0}$ be a general fermionic Fock state, parametrized as
\begin{equation} \label{eq:A20}
\ket{\mathbf{x}}_{0}=f^{\dagger x_{1}}_{1}\cdots{f}^{\dagger x_{m}}_{m}\ket{0}_{0}=[\mathbf{x}]_{0}\ket{0}_{0},
\end{equation}
where $\mathbf{x}=(x_{1},...,x_{m})$, and each $x_{i}$ is the occupation number of operator $n_{0,i}$. Then, the state
\begin{equation}
\ket{\mathbf{x}}_{\varphi}=[\mathbf{x}]_{\varphi}\ket{0}_{\varphi},
\end{equation}
is a fermionic-anyon Fock state with the same occupation numbers.
\end{prop}

\begin{proof}
We have that
\begin{align*}
[\mathbf{x}]_{\varphi}&=J^{-1}_{\varphi}(f^{\dagger x_{1}}_{1}\cdots{f}^{\dagger x_{i}}_{i}\cdots{f}^{\dagger x_{m}}_{m})\\&=\left(a^{\dagger x_{1}}_{\varphi,1}\right)\cdots\left(a^{\dagger x_{i}}_{\varphi,i}e^{ix_{i}\varphi\sum^{i-1}_{k=1}n_{\varphi,k}}\right)\cdots\left(a^{\dagger x_{m}}_{\varphi,m}e^{ix_{m}\varphi\sum^{m-1}_{k=1}n_{\varphi,k}}\right)
\\&=a^{\dagger x_{1}}_{\varphi,1}\cdots{a}^{\dagger x_{m}}_{\varphi,m}\left(e^{i\varphi\sum_{k=1}^{m-1}\left(\sum_{l=k+1}^{m}(x_{l})\right)n_{\varphi,k}}\right).
\end{align*}
Acting with $[x]_{\varphi}$ over the fermionic-anyon vacuum $\ket{0}_{\varphi}$ then gives us
\begin{align*}
[\mathbf{x}]_{\varphi}\ket{0}_{\varphi}&=a^{\dagger x_{1}}_{\varphi,1}\cdots{a}^{\dagger x_{m}}_{\varphi,m}\left(e^{i\varphi\sum_{k=1}^{m-1}\left(\sum_{l=k+1}^{m}(x_{l})\right)n_{\varphi,k}}\right)\ket{0}_{\varphi}\\&=a^{\dagger x_{1}}_{\varphi,1}\cdots{a}^{\dagger x_{m}}_{\varphi,m}\ket{0}_{\varphi}.
\end{align*}
Since each operator $a^{\dagger}_{\varphi,i}$ is a creation operator for $n_{\varphi,i}$, the state in the righthand side of the last equation is a fermionic-anyon Fock state with occupation numbers $\mathbf{x}=(x_{1},...,x_{m})$.
\end{proof}

Note that if the operators are not acting on the vacuum state in decreasing order, as shown in Eq.~\eqref{eq:A20}, some phases can appear. For example, consider the state $f^{\dagger }_{5}{f}^{\dagger}_{3}\ket{0}_{0}$. Applying $J^{-1}_{\varphi}$ on the creation operators leads to $a^{\dagger}_{\varphi,5}e^{(i\varphi\sum_{i=1}^{4})}{a}^{\dagger}_{\varphi,3}e^{(i\varphi\sum_{i=1}^{2})}\ket{0}_{\varphi}=e^{i\varphi}a^{\dagger}_{\varphi,5}{a}^{\dagger}_{\varphi,3}\ket{0}_{\varphi}$. This implies that even though $J^{-1}_{\varphi}$ preserves fermionic commutation relations, its action over fermionic operators that create physical states gives us fermionic-anyon states with the correct behaviour under particle exchange, as proven in the corollary bellow.

\begin{coro}\label{cor: ferm-any-stat-equiv}
For any fermionic state of the form
\begin{equation}
\ket{\psi}_{0}=\sum_{\mathbf{x}\in\{0,1\}^{m}}\psi_{\mathbf{x}}f^{\dagger x_{1}}_{1}\cdots f^{\dagger x_{m}}_{m}\ket{0}_{0}=[\psi]_{0}\ket{0}_{0},
\end{equation}
we have that
\begin{equation}
\ket{\psi}_{\varphi}=[\psi]_{\varphi}\ket{0}_{\varphi},
\end{equation}
is a fermionic-anyon state that satisfies
\begin{equation}
\braket{\mathbf{x}}{\psi}_{\varphi}=\braket{\mathbf{x}}{\psi}_{0}=\psi_{\mathbf{x}},
\end{equation}
for all $\varphi\in[0,2\pi)$ and all $\mathbf{x}\in\{0,1\}^{m}$.
\end{coro}
%\td{Here, again, if creation operators are not in decreasing order, some phases can appear in the projection of the state onto a given basis element.}

\begin{proof}
Just apply $J^{-1}_{\varphi}$ to $[\psi]_{0}$ and use the linearity $J^{-1}_{\varphi}$ of together with the result of proposition 5.
\end{proof}

This last result shows us that we can specify any fermionic-anyon state in the Fock basis, in terms of a fermionic state in the Fock basis, giving us the hint that we might be able to translate any construction used to study fermionic states to obtain results about the structure of fermionic-anyons states. In the next section, we discuss how this change of basis interacts with transformations that preserve commutation relations inside the same algebra.

\section{Fermionic-anyon separability}

Here we use the machinery developed in the last section to describe the behaviour of fermionic-anyon states under local changes of basis and show how it is used to define particle entanglement, by proving Theorem \ref{theorem: SM_SD for anyons}. First, we define what it means for a map over an *-algebra to be a canonical transformation, and describe local changes of basis in fermionic systems as a special type of canonical transformation. Next, we show how to build the local basis changes for fermionic anyons using the fractional Jordan-Wigner transform. Lastly, we show how to use these transformations to prove the existence of a fermionic-anyon Slater decomposition.

\subsection{Canonical transformations}

We call all maps that preserve the commutation relations of a fermionic-anyon algebra as canonical transformations, as in the following definition.

\begin{defi}
Let $\mathcal{A}^{\varphi}_{m}$ and let $g:\mathcal{A}^{\varphi}_{m}\rightarrow\mathcal{A}^{\varphi}_{m}$ be a function. We say that $g$ is a \textbf{canonical transformation} if and only if
\begin{align}
g(a_{\varphi,i})g(a_{\varphi,j}^{\dagger})+e^{-i\varphi\epsilon_{ij}}g(a_{\varphi,j}^{\dagger})g(a_{\varphi,i})&=\delta_{i,j},\\
g(a_{\varphi,i})g(a_{\varphi,j})+e^{i\varphi\epsilon_{ij}}g(a_{\varphi,j})g(a_{\varphi,i})&=0, 
\end{align}
\end{defi}

In other words, a canonical transformation is any map that transforms the set of *-algebra's generators into another set of generators. This implies that canonical transformations act as a change of variables in the description of a system of indistinguishable particles. Note that even though the statistical transmutation maps also preserve commutation relations, they just translate the observables of a particle system into a system with a different kind of particle, which is why we insist that canonical transformations need to be defined by endomorphisms, that is functions from an *-algebra to itself. Given that the main characteristic of a canonical transformation is to preserve a particular algebraic relation, it should be no surprise that the following is true.

\begin{prop}
All *-algebra homomorphisms from $\mathcal{A}^{\varphi}_{m}$ to itself are canonical transformations.
\end{prop}

\begin{proof}
Just apply corollary 1 to the fermionic-anyon commutation relations.
\end{proof}

Therefore, we can look for transformations representing local basis changes among *-algebra homomorphisms. For the specific case of $\mathcal{F}_{m}$, local basis changes are represented by a special subset of known canonical transformations, which are defined below.

\begin{prop}
The *-algebra homomorphism $B_{U,V}:\mathcal{F}_{m}\rightarrow\mathcal{F}_{m}$ of the form
\begin{equation}
B_{U,V}(f_{i})=\sum^{m}_{j=1}U_{i,j}f_{j}+\sum^{m}_{k=1}V_{i,k}f^{\dagger}_{k},
\end{equation}
where $U$ and $V$ are matrices such that
\begin{align}
&UU^{\dagger}+VV^{\dagger}=\mathds{1},\\
&UV^{\intercal}+VU^{\intercal}=0,
\end{align}
is well-defined, and is called a \textbf{multi-mode fermionic Bogoliubov transformation}. When $V=0$, $B_{U,0}$ is called a \textbf{local change of basis}.
\end{prop}

\begin{proof}
To prove that $B_{U,V}$ is well defined, we just need to compute the canonical anticommutation relations and show that they are preserved. First, notice that
\begin{equation}
\begin{split}
&\{B_{U,V}(f_{i});B_{U,V}(f^{\dagger}_{j})\}=\left\{\sum^{m}_{k=1}U_{i,k}f_{k}+\sum^{m}_{l=1}V_{i,l}f^{\dagger}_{l};\sum^{m}_{p=1}\Bar{V}_{j,p}f_{p}+\sum^{m}_{q=1}\Bar{U}_{j,q}f^{\dagger}_{q}\right\}=\\
=&\sum^{m}_{k,p=1}U_{i,k}\Bar{V}_{j,p}\{f_{k};f_{p}\}+\sum^{m}_{k,q=1}U_{i,k}\Bar{U}_{j,q}\{f_{k};f^{\dagger}_{q}\}+\sum^{m}_{l,p=1}V_{i,l}\Bar{V}_{j,p}\{f^{\dagger}_{l};f_{p}\}+\sum^{m}_{l,q=1}V_{i,l}\Bar{U}_{j,q}\{f^{\dagger}_{l};f^{\dagger}_{q}\}=\\
=&\sum^{m}_{k,q=1}U_{i,k}\Bar{U}_{j,q}\delta_{k,q}+\sum^{m}_{l,p=1}V_{i,l}\Bar{V}_{j,p}\delta_{l,p}=\sum^{m}_{k=1}U_{i,k}\Bar{U}_{j,k}+\sum^{m}_{l=1}V_{i,l}\Bar{V}_{j,l}=(UU^{\dagger}+VV^{\dagger})_{i,j}=\delta_{i,j}.
\end{split}
\end{equation}
Lastly, we have that
\begin{equation}
\begin{split}
&\{B_{U,V}(f_{i});B_{U,V}(f_{j})\}=\left\{\sum^{m}_{k=1}U_{i,k}f_{k}+\sum^{m}_{l=1}V_{i,l}f^{\dagger}_{l};\sum^{m}_{p=1}U_{j,p}f_{p}+\sum^{m}_{q=1}V_{j,q}f^{\dagger}_{q}\right\}=\\
=&\sum^{m}_{k,p=1}U_{i,k}U_{j,p}\{f_{k};f_{p}\}+\sum^{m}_{k,q=1}U_{i,k}V_{j,q}\{f_{k};f^{\dagger}_{q}\}+\sum^{m}_{l,p=1}V_{i,l}U_{j,p}\{f^{\dagger}_{l};f_{p}\}+\sum^{m}_{l,q=1}V_{i,l}V_{j,q}\{f^{\dagger}_{l};f^{\dagger}_{q}\}=\\
=&\sum^{m}_{k,q=1}U_{i,k}V_{j,q}\delta_{k,q}+\sum^{m}_{l,p=1}V_{i,l}U_{j,p}\delta_{l,p}=\sum^{m}_{k=1}U_{i,k}V_{j,k}+\sum^{m}_{l=1}V_{i,l}U_{j,l}=(UV^{\intercal}+VU^{\intercal})_{i,j}=0.
\end{split}
\end{equation}
\end{proof}

These transformations are not just canonical, but they also send generators into linear combinations of generators and their conjugates. This implies that local basis changes for fermions just redefine creation operators as linear combinations of the originals, which is exactly what happens when redefining the basis of states occupied by a single particle system. This behaviour is what we aim for when looking for a fermionic-anyon local basis change in the next section.

\subsection{Changes-of-basis induced by fractional Jordan-Wigner transforms}

From the characterisation of *-algebra homomorphisms as canonical transformations, we find that there is a natural way to define local basis changes for fermionic-anyons with any exchange parameter $\varphi$. 

\begin{defi}\label{def: bogo_ferm_any}
Given a Bogoliubov transformation for fermions $B_{U,V}:\mathcal{F}_{m}\rightarrow\mathcal{F}_{m}$ , we call the *-algebra homomorphism $B^{\varphi}_{U,V}:\mathcal{A}^{\varphi}_{m}\rightarrow\mathcal{A}^{\varphi}_{m}$ defined by
\begin{equation}
B^{\varphi}_{U,V}=\left(J^{-1}_{\varphi}\circ B_{U,V}\circ J_{\varphi}\right)    
\end{equation}
an \textbf{induced Bogoliubov transformation} over $\mathcal{A}^{\varphi}_{m}$. When $V=0$, we call it an \textbf{induced local change of basis}.
\end{defi}

Note that this is completely analogous to the similarity transformation between linear operators induced by a basis change in the target vector space. Our observations about the exchange transmutation maps being akin to basis changes are thus made completely clear. The next task is to compute the action of fermionic-anyon Bogoliubov transformations over generators

\begin{prop}
The map $B^{\varphi}_{U,V}:\mathcal{A}^{\varphi}_{m}\rightarrow\mathcal{A}^{\varphi}_{m}$ acts over generators as
\begin{equation}
B^{\varphi}_{U,V}(a_{\varphi,i})=\left(\sum^{m}_{j=1}U_{i,j}a_{\varphi,j}e^{i\varphi\sum^{j-1}_{l=1}n_{\varphi,l}}+\sum^{m}_{k=1}V_{i,k}a^{\dagger}_{\varphi,k}e^{-i\varphi\sum^{k-1}_{l=1}n_{\varphi,l}}\right)e^{-i\varphi\sum^{i-1}_{q=1}\Bar{n}_{\varphi,q}},
\end{equation}
where
\begin{align*}
\Bar{n}_{\varphi,q}&=[B_{U,V}(n_{0,q})]_{\varphi}\\&=\left[\sum^{m}_{j_{1},j_{2}=1}U_{q,j_{1}}U^{*}_{q,j_{2}}f_{j_{1}}f^{\dagger}_{j_{2}}+\sum^{m}_{j_{1},j_{2}=1}U_{q,j_{1}}V^{*}_{q,j_{2}}f_{j_{1}}f_{j_{2}}+\sum^{m}_{j_{1},j_{2}=1}V_{q,j_{1}}U^{*}_{q,j_{2}}f^{\dagger}_{j_{1}}f^{\dagger}_{j_{2}}+\sum^{m}_{j_{1},j_{2}=1}V_{q,j_{1}}V^{*}_{q,j_{2}}f^{\dagger}_{j_{1}}f_{j_{2}}\right]_{\varphi}
\end{align*}
\end{prop}

\begin{proof}
Let's compute $B^{\varphi}_{U,V}(a_{\varphi,i})$. 
\begin{align*}
B^{\varphi}_{U,V}(a_{\varphi,i})&=\left(J^{-1}_{\varphi}\circ B_{U,V}\circ J_{\varphi}\right)(a_{\varphi,i})\\&=\left(J^{-1}_{\varphi}\circ B_{U,V}\right)\left(f_{i}e^{-i\varphi\sum^{i-1}_{q=1}n_{0,q}}\right)\\&=J^{-1}_{\varphi}\left(B_{U,V}(f_{i})B_{U,V}\left(e^{-i\varphi\sum^{i-1}_{q=1}n_{0,q}}\right)\right)\\
&=J^{-1}_{\varphi}\left(B_{U,V}(f_{i})e^{-i\varphi\sum^{i-1}_{q=1}B_{U,V}(n_{0,q})}\right)\\&=J^{-1}_{\varphi}\left(\sum^{m}_{j=1}U_{i,j}f_{j}+\sum^{m}_{k=1}V_{i,k}f^{\dagger}_{k}\right)\left(e^{-i\varphi\sum^{i-1}_{q=1}J^{-1}_{\varphi}(B_{U,V}(n_{0,q}))}\right)\\
&=\left(\sum^{m}_{j=1}U_{i,j}[f_{j}]_{\varphi}+\sum^{m}_{k=1}V_{i,k}[f^{\dagger}_{k}]_{\varphi}\right)e^{-i\varphi\sum^{i-1}_{q=1}[B_{U,V}(n_{0,q})]_{\varphi}}\\&=\left(\sum^{m}_{j=1}U_{i,j}a_{\varphi,j}e^{i\varphi\sum^{j-1}_{l=1}n_{\varphi,l}}+\sum^{m}_{k=1}V_{i,k}a^{\dagger}_{\varphi,k}e^{-i\varphi\sum^{j-1}_{l=1}n_{\varphi,l}}\right)e^{-i\varphi\sum^{i-1}_{q=1}\Bar{n}_{\varphi,q}}
\end{align*}
\end{proof}

It is perplexing that a transformation so complex is the correct transformation for implementing a local basis change for fermionic anyons. Nevertheless, by all the considerations made up to this moment, it is necessarily the case $B^{\varphi}_{U,0}$ induces a map that acts over single-particle fermionic-anyon states in the same way as $B_{U,0}$ does for fermionic states, implying that these maps have the same physical interpretation acting over the Fock-space. With all of these results in place, we are now able to prove our main theorem.

\subsection{Single-particle entanglement for fermionic-anyons}\label{app: single particle entanglement}

First, let us restate the definition of separability for fermions and fermionic anyons in terms of the more general notation we have been using so far.

\begin{defi}
Let $\ket{\psi}_{0}\in\mathcal{H}_{\mathcal{F}_{m}}$ be a pure fermionic state. We say that $\ket{\psi}_{0}$ is separable, if there is a local change of basis $B_{U,0}:\mathcal{F}_{m}\rightarrow\mathcal{F}_{m}$ such that
\begin{equation}\label{eq: fer_sep}
B_{U,0}([\psi]_{0})=[\mathbf{x}]_{0},
\end{equation}
with $[\mathbf{x}]_{0}=f^{\dagger x_{1}}_{1}\cdots f^{\dagger x_{m}}_{m}$, for some $\mathbf{x}\in\{0,1\}^{m}$.
\end{defi}

This definition implies capturing the notion that a state with no ``particle to particle" correlations should be locally equivalent to some Fock basis state, which is a state where each mode has a well-defined occupation number. Now, using our machinery, we are in a position to state the definition of fermionic-anyon separability.

\begin{defi}
Let $\ket{\phi}_{\varphi}\in\mathcal{H}_{\mathcal{A}^{\varphi}_{m}}$ be a pure fermionic-anyon state. We say that $\ket{\phi}_{\varphi}$ is separable, if there is a local change of basis $B^{\varphi}_{U,0}:\mathcal{A}^{\varphi}_{m}\rightarrow\mathcal{A}^{\varphi}_{m}$ such that
\begin{equation}
B^{\varphi}_{U,0}([\phi]_{\varphi})=[\mathbf{x}]_{\varphi},
\end{equation}
with $[\mathbf{x}]_{0}=a^{\dagger x_{1}}_{\varphi,1}\cdots a^{\dagger x_{m}}_{\varphi,m}$, for some $\mathbf{x}\in\{0,1\}^{m}$.
\end{defi}

This definition comes from just applying the fractional Jordan-Wigner transform on both sides of the Eq.~\eqref{eq: fer_sep}, and using the $\varphi$-representation of the operators creating the states. 

{\begin{defi}[Single-particle entanglement for fermionic \cite{gigena2015}]
    For a $N$-fermionic system with $m$-modes as shown in Eq.\eqref{cor: ferm-any-stat-equiv}, the single-particle entanglement can be quantified through the minimization over all the possible mode representations 
\begin{equation}\label{eq: min_ent}
    E_{SP}(\ket{\psi}) = \min_{f}\sum_i H(\langle f_i^{\dagger}f_i\rangle,\langle f_if_i^{\dagger}\rangle ),
\end{equation}
where $H(p,1-p)=-p\log p - (1-p)\log (1-p)$ is the Shannon binary entropy, and $f_i$ are the fermionic operators, transformed accordingly to the Bogoliubov transformation.
\end{defi}}

{\begin{teo}[Single-particle entanglement for fermionic-anyons]
   For a fermionic-anyons state with a fixed number of particles, there exists a mode representation such that the single particle of fermionic-anyons has the same eigenvalues as the mapped fermionic single-particle state, that minimizes Eq.\eqref{eq: min_ent}
\end{teo}}

{\begin{proof}
For a $N$ particles fermionic-anyon state $\ket{\psi}_{\varphi}$, there exist a $N$-fermionic state $\ket{\psi}_{0} = J_{\varphi}[\ket{\psi}_{\varphi}]$. For a given fermionic mode representation described by the anti-symmetric operator $\{f_k^{\dagger},f_k \}_k$, the single-particle fermionic state $\rho_{\text{sp}}^{(0)}$ with elements
\begin{eqnarray}
    (\rho_{\text{sp}}^{(0)})_{kl} = \langle f_l^{\dagger} f_k \rangle\subtiny{0}{0}{\ket{\psi}_{0}}.
\end{eqnarray}
Moreover, there exists a mode representation $\{\tilde{f}_k^{\dagger},\tilde{f}_k \}_k$, that minimizes the entropy in Eq.\eqref{eq: min_ent}, and diagonalizes $\rho_{\text{sp}}^{(0)}$ \cite{gigena2015}. It is connected with $\{f_k^{\dagger},f_k \}_k$ by a Bogouliubov transformation $U$, such that $\tilde{f}_k^{\dagger} = U {f}_k^{\dagger} U^{\dagger}$, then
\begin{equation}
    (\tilde{rho}_{\text{sp}}^{(0)})_{kl} = (U\rho_{\text{sp}}^{(0)}U^{\dagger})_{kl} = \langle \tilde{f}_l^{\dagger} \tilde{f}_k \rangle\subtiny{0}{0}{\ket{\psi}_{0}}\delta_{k,l} = \omega_k \delta_{k,l},
\end{equation}
where $\omega_k$ are the eigenvalues of single-particle state. The Shannon entropy of $\{\omega_k\}_k$ coincides with the von Neumann entropy of $\rho_{\text{sp}}^{(0)}$, as expected. On the other hand,  as a consequence of Corollary.\ref{cor: ferm-any-stat-equiv}, the operator number $N_k = f_k^{\dagger} f_k$ is invariant under JWT. Therefore, there exists a fermionic-anyon mode representation $\{\tilde{a}_k^{\dagger},\tilde{a}_k \}_k$ such that
\begin{equation}
    \langle \tilde{a}_k^{\dagger}\tilde{a}_k\rangle\subtiny{0}{0}{\ket{\psi}_{\varphi}}  = \langle \tilde{f}_k^{\dagger} \tilde{f}_k \rangle\subtiny{0}{0}{\ket{\psi}_{0}},
\end{equation}
for $\tilde{f}_k = J_{\varphi}[\tilde{a}_k]$. It implies that there exists a single particle mode representation that reflects the entanglement between the particles, despite the fact that the von Neumann entropy of the single particle, obtained by means of the partial trace in any basis, does not represent the separability of the fermionic-anyons system. 
\end{proof}}

\subsection{Fermionic-anyon Slater decomposition}\label{app: slater decomposition}{app: single particle entanglement}
In the special case of general two-particle fermionic states, we are not just able to describe separability, but also able to define a normal form.
\begin{defi}
Let $\ket{\psi}_{0}$ be a state of the form
\begin{equation}
\ket{\psi}_{0}=\frac{1}{2}\sum_{i<j}c_{i,j}f^{\dagger}_{i}f^{\dagger}_{j}\ket{0}_{0}=\sum^{m}_{i,j=1}v_{i,j}f^{\dagger}_{i}f^{\dagger}_{j}\ket{0}_{0},
\end{equation}
where $v_{j,i}=-v_{i,j}$, $\sum^{m}_{i,j=1}\Bar{v}_{i,j}v_{j,i}=-1/2$ and $v_{i,j}=1/2(c_{i,j})$ for $i<j$. Then, the \textit{Slater decomposition} of $\ket{\psi}_{0}$ has the form
\begin{equation}
\ket{\psi}_{0}=\sum^{\mu}_{k=1}z_{k}\Tilde{f}^{\dagger}_{2k-1}\Tilde{f}^{\dagger}_{2k}\ket{0}_{0},
\end{equation}
where $B_{U,0}(\Tilde{f}^{\dagger}_{i})=f^{\dagger}_{i}$ with $U$ being such that $UVU^{\intercal}=Z$, where $V$ is the antisymmetric matrix with coefficients $v_{i,j}$ and $Z=\diag{Z_{0},Z_{1},...,Z_{\mu}}$ where
\begin{equation}
Z_{0}=0\text{, and }Z_{i}=\begin{bmatrix}
0 & z_{i}\\
-z_{i} & 0
\end{bmatrix}
\end{equation}
\end{defi}

We now show that this normal form also exists for general two-particle fermionic anyons states and that the coefficients of the expansion are the same as the analogous expansion for a fermionic state with the same initial amplitudes. In other words, now it is time to prove the following Theorem.
\begin{teo}[Schmidt decomposition for fermionic anyons]\label{theorem: SM_SD for anyons}
Any pure state of two fermionic anyons with a fixed number of modes has a Schmidt decomposition with the same expansion coefficients as its Schliemann fermionic state counterpart.
\end{teo}
\begin{proof}
First, let $\ket{\psi}_{\varphi}$ be a general fermionic anyon two-particle state as bellow
\begin{equation}
\ket{\psi}_{\varphi}=\sum^{m}_{i,j=1}v_{i,j}a^{\dagger}_{\varphi,i}a^{\dagger}_{\varphi,j}\ket{0}_{\varphi}.
\end{equation}
We have that,
\begin{align*}
\sum^{m}_{i,j=1}v_{i,j}a^{\dagger}_{\varphi,i}a^{\dagger}_{\varphi,j}\ket{0}_{\varphi}&=\sum^{m}_{i,j=1}v_{i,j}[f^{\dagger}_{i}f^{\dagger}_{j}]_{\varphi}\ket{0}_{\varphi}\\&=\sum^{m}_{i,j=1}v_{i,j}[B_{U,0}(\Tilde{f}^{\dagger}_{k}\Tilde{f}^{\dagger}_{l})]_{\varphi}\ket{0}_{\varphi}\\&=\sum^{m}_{i,j=1}v_{i,j}\sum^{m}_{k,l=1}U_{i,k}U_{j,l}[\Tilde{f}^{\dagger}_{k}\Tilde{f}^{\dagger}_{l}]_{\varphi}\ket{0}_{\varphi}\\&
=\sum^{m}_{k,l=1}(UVU^{\intercal})_{k,l}[\Tilde{f}^{\dagger}_{k}\Tilde{f}^{\dagger}_{l}]_{\varphi}\ket{0}_{\varphi}\\&=\sum^{m}_{k,l=1}(Z)_{k,l}[\Tilde{f}^{\dagger}_{k}\Tilde{f}^{\dagger}_{l}]_{\varphi}\ket{0}_{\varphi}\\&=\sum^{\mu}_{k=1}z_{k}[\Tilde{f}^{\dagger}_{2k-1}\Tilde{f}^{\dagger}_{2k}]_{\varphi}\ket{0}_{\varphi}.
\end{align*}
Now, consider that
\begin{align*}
[\Tilde{f}^{\dagger}_{2k-1}\Tilde{f}^{\dagger}_{2k}]_{\varphi}&=[B_{U^{\dagger},0}(f^{\dagger}_{2k-1}f^{\dagger}_{2k})]_{\varphi}=J^{-1}_{\varphi}(B_{U^{\dagger},0}(f^{\dagger}_{2k-1}f^{\dagger}_{2k}))\\&=B^{\varphi}_{U^{\dagger},0}\left(a^{\dagger}_{\varphi,2k-1}a^{\dagger}_{\varphi,2k}e^{-i\varphi(n_{\varphi,2k-1}+2\sum^{2k-2}_{l=1}n_{\varphi,l})}\right)\\
&=B^{\varphi}_{U^{\dagger},0}(a^{\dagger}_{\varphi,2k-1})B^{\varphi}_{U^{\dagger},0}(a^{\dagger}_{\varphi,2k})B^{\varphi}_{U^{\dagger},0}\left(e^{-i\varphi(n_{\varphi,2k-1}+2\sum^{2k-2}_{l=1}n_{\varphi,l})}\right)\\&=\Tilde{a}^{\dagger}_{\varphi,2k-1}\Tilde{a}^{\dagger}_{\varphi,2k}e^{-i\varphi(\Bar{n}_{\varphi,2k-1}+2\sum^{2k-2}_{l=1}\Bar{n}_{\varphi,l})},
\end{align*}
where we have defined $B^{\varphi}_{U,0}(\Tilde{a}^{\dagger}_{\varphi,i})=a^{\dagger}_{\varphi,i}$, and $\Bar{n}_{\varphi,i}=\Tilde{a}^{\dagger}_{\varphi,i}\Tilde{a}_{\varphi,i}$, for all $i$. Thefore, we have that
\begin{align*}
\ket{\psi}_{\varphi}&=\sum^{m}_{i,j=1}v_{i,j}a^{\dagger}_{\varphi,i}a^{\dagger}_{\varphi,j}\ket{0}_{\varphi}\\&=\sum^{\mu}_{k=1}z_{k}\Tilde{a}^{\dagger}_{\varphi,2k-1}\Tilde{a}^{\dagger}_{\varphi,2k}e^{-i\varphi(\Bar{n}_{\varphi,2k-1}+2\sum^{2k-2}_{l=1}\Bar{n}_{\varphi,l})}\ket{0}_{\varphi}\\&=\sum^{\mu}_{k=1}z_{k}\Tilde{a}^{\dagger}_{\varphi,2k-1}\Tilde{a}^{\dagger}_{\varphi,2k}\ket{0}_{\varphi},
\end{align*}
where we used that
\begin{equation}
\ket{0}_{\varphi} = e^{-i\varphi(\Bar{n}_{\varphi,2k-1}+2\sum^{2k-2}_{l=1}\Bar{n}_{\varphi,l})}\ket{0}_{\varphi},
\end{equation}
which is true since $B^{\varphi}_{U,0}$ maps annihilation operators to annihilation operators, thus they still annihilate the fermionic-anyon vacuum.
\end{proof}

\section{Fermionic linear optics}

{Here we introduce the set of fermionic dynamics called fermionic linear optics, known to be classically simulable \cite{terhal_classical_2002}. Then we describe the set of fermionic operators that generate fermionic linear optics and give a special generating set that involves only nearest-neighbour operators.}
We begin by introducing the algebra of quadratic fermionic operators.

\begin{prop}
The quadratic fermionic monomials $f^{\dagger}_{i}f_{j}$, $f^{\dagger}_{i}f^{\dagger}_{j}$ and $f_{i}f_{j}$, together with the identity operator $\mathds{1}$, generate a Lie-algebra with commutation relations given by
\begin{align}
[f_{i}f_{j};f_{k}f_{l}]&=[f^{\dagger}_{i}f^{\dagger}_{j};f^{\dagger}_{k}f^{\dagger}_{l}]=0\\
[f^{\dagger}_{i}f_{j};f^{\dagger}_{k}f^{\dagger}_{l}]&=\delta_{j,l}f^{\dagger}_{i}f^{\dagger}_{k}\\
[f^{\dagger}_{i}f_{j};f^{\dagger}_{k}f_{l}]&=\delta_{j,k}f^{\dagger}_{i}f_{l}-\delta_{i,l}f^{\dagger}_{k}f_{j}\\
[f^{\dagger}_{i}f^{\dagger}_{j};f_{k}f_{l}]&=(\delta_{k,i}\delta_{l,j}-\delta_{l,i}\delta_{k,j})\mathds{1}+(\delta_{j,k}f^{\dagger}_{i}f_{l}+\delta_{l,i}f^{\dagger}_{j}f_{k})-(\delta_{k,i}f^{\dagger}_{j}f_{l}+\delta_{l,j}f^{\dagger}_{i}f_{k})
\end{align}
\end{prop}

\begin{proof}
This is proven by straightforward calculation using the canonical anticommutation relations.
\end{proof}

\begin{coro}
The sub-algebra generated by the fermionic number operators $n_{0,i}$, is the maximal commuting sub-algebra of the Lie-Algebra of quadratic fermionic monomials. The commutation relations
\begin{subequations}
\begin{align}
[n_{0,i};f^{\dagger}_{k}f_{l}]&=(\delta_{i,k}-\delta_{i,l})f^{\dagger}_{k}f_{l}\\
[n_{0,i};f_{k}f_{l}]&=-(\delta_{i,k}+\delta_{i,l})f_{k}f_{l}\\
[n_{0,i};f^{\dagger}_{k}f^{\dagger}_{l}]&=(\delta_{i,k}+\delta_{i,l})f^{\dagger}_{k}f^{\dagger}_{l}
\end{align}
\end{subequations}
imply that this sub-algebra is the $\mathfrak{so}_{2m}$ Lie-algebra.
\end{coro}

\begin{proof}
See Chapter 21 of Classical Groups for Physicists  \cite{wybourne1974}.
\end{proof}

Given the above relations, we are now in order to define the fermionic linear-optical operators, as follows.
\begin{defi}
The Lie-group generated by taking the exponential of the hermitian operators
\begin{equation}
\left\{n_{0,i},\frac{1}{2}(f^{\dagger}_{i}f_{j}+f^{\dagger}_{j}f_{i}),\frac{i}{2}(f^{\dagger}_{i}f_{j}-f^{\dagger}_{j}f_{i}),\frac{1}{2}(f^{\dagger}_{i}f^{\dagger}_{j}+f_{j}f_{i}),\frac{i}{2}(f^{\dagger}_{i}f^{\dagger}_{j}-f_{j}f_{i})|i,j=1,\dots,m\right\},
\end{equation}
is called the group of fermionic linear-optical operators, or \textbf{FLO} operators for short, which is isomorphic to $U(1)\cross SO(2M)$. The group elements
\begin{align}
PS_{i}(\phi)&=\exp[i\phi(f^{\dagger}_{i}f_{i})]\\
BS_{i,j}(\theta)&=\exp[i\theta(f^{\dagger}_{i}f_{j}+f^{\dagger}_{j}f_{i})]\\
PA_{i,j}(\nu)&=\exp[i\nu(f^{\dagger}_{i}f^{\dagger}_{j}+f_{j}f_{i})],
\end{align}
are generators of the FLO group, and are called by special names. The operator $PS_{i}(\nu)$ is called a \textit{phase-shifter}, $BS_{i,j}(\theta)$ is called a \textit{beam-splitter}, and $PA_{i,j}(\nu)$ is called a \textit{parametric amplifier}. 
\end{defi}

{Having characterised the generators of the fermionic linear-optical dynamics, we now look for the smallest possible description of this set. In order to do this, we need to define another special element.}

\begin{defi}
The \textbf{fermionic swap} is the fermionic linear-optical element given by
\begin{equation}
f\textsc{SWAP}_{i,j}=\exp\left[i\frac{\pi}{2}(f^{\dagger}_{i}-f^{\dagger}_{j})(f_{i}-f_{j})\right]=\mathds{1}-(n_{0,i}+n_{0,j})+(f^{\dagger}_{i}f_{j}+f^{\dagger}_{j}f_{i}),
\end{equation}
and its action over creation operators is
\begin{equation}
(f\textsc{SWAP}_{i,j})f^{\dagger}_{l}(f\textsc{SWAP}_{i,j})=\begin{cases}
f^{\dagger}_{j}&\text{, if }l=i,\\
f^{\dagger}_{i}&\text{, if }l=j,\\
f^{\dagger}_{l}&\text{, if }l\neq i,j.
\end{cases}
\end{equation}
\end{defi}

{Using this definition, we can provide a simple description of fermionic linear optics.}

\begin{prop}
The set of fermionic linear-optical elements
\begin{equation}
\left\{PS_{1}(\phi),BS_{1,2}(\theta),PA_{1,2}(\nu),f\textsc{SWAP}_{i,i+1}|i=1,\dots,m\right\}
\end{equation}
generates the FLO group.
\end{prop}

\begin{proof}
We prove this by showing how to build any distant-mode phase-shifter, beam-splitter and parametric amplifier. First, notice that
\begin{equation}
f\textsc{SWAP}_{i,i+1}f\textsc{SWAP}_{i+1,i+2}f\textsc{SWAP}_{i,i+1}=f\textsc{SWAP}_{i,i+2},
\end{equation}
for any $i=1,\dots,m$. This implies that for $i<j$ with $i,j=1,\dots,m$,
\begin{equation}
f\textsc{SWAP}_{i,j}=\underbrace{f\textsc{SWAP}_{i,i+1}\cdots f\textsc{SWAP}_{j-2,j-1}}_\text{$(j-i-1)$ times}f\textsc{SWAP}_{j-1,j}\underbrace{f\textsc{SWAP}_{j-2,j-1}\cdots f\textsc{SWAP}_{i,i+1}}_\text{$(j-i-1)$ times}.
\end{equation}
Therefore, we can built any fermionic swap using only nearest-neighbour fermionic swaps. This also implies that we can build any other phase-shifter, beam-splitter or parametric amplifier since
\begin{align*}
PS_{i}(\phi)&=f\textsc{SWAP}_{1,i}PS_{1}(\phi)f\textsc{SWAP}_{1,i},\\
BS_{i,j}(\theta)&=f\textsc{SWAP}_{1,i}f\textsc{SWAP}_{2,j}BS_{1,2}(\theta)f\textsc{SWAP}_{2,j}f\textsc{SWAP}_{1,i},\\
PA_{i,j}(\nu)&=f\textsc{SWAP}_{1,i}f\textsc{SWAP}_{2,j}PA_{1,2}(\nu)f\textsc{SWAP}_{2,j}f\textsc{SWAP}_{1,i},   
\end{align*}
which proves the proposition.
\end{proof}

\section{Fermionic-anyonic linear optics and matchgates}\label{app: fermionic anyons LO}

{Having characterised fermionic linear-optical dynamics, we now move on to their anyonic generalisation and the connection to matchgates, a special class of quantum circuits. This is done by studying which fermionic linear-optical operators have an invariant form under the fractional Jordan-Wigner homomorphism.}

\begin{defi}
The set of fermionic anyon operators generated by finite products of
\begin{align}
PS_{i}(\theta)&=\exp[i\theta(a^{\dagger}_{\varphi,i}a_{\varphi,i})]\\
BS_{i,j}(\theta)&=\exp[i\theta(a^{\dagger}_{\varphi,i}a_{\varphi,j}+a^{\dagger}_{\varphi,j}a_{\varphi,i})]\\
PA_{i,j}(\theta)&=\exp[i\theta(a^{\dagger}_{\varphi,i}a^{\dagger}_{\varphi,j}+a_{\varphi,j}a_{\varphi,i})],
\end{align}
is called the set of fermionic-anyon linear-optical dynamics, or $\mathbf{\Phi LO}$ for short. For the particular case of $\varphi=\pi$, it is called the set of qubit linear-optical dynamics, or $\mathbf{QLO}$, since for $\varphi=\pi$ the fermionic-anyons creation and annihilation operators behave like spin-1/2 raising and lowering operators. In particular, $\mathbf{QLO}$ coincides with what is called in the literature by \textbf{matchgate circuits}.
\end{defi}

\section{Proof that Bogoliubov transformations for fermionic anyons are classically simulable}\label{app: classically simulable}

The proof that Bogoliubov transformations for anyons are classically simulable is as follows. For any state in the fermionic-anyon Fock basis $\ket{\vb{x}}_{a}$, we know that
\begin{equation}
\ket{\vb{x}}_{a}=J^{-1}(\ket{\vb{x}}_{f}),
\end{equation}
where $\ket{\vb{x}}_{f}$ is a state in the fermionic Fock-basis with the same occupation numbers. Since $J^{-1}$ is a *- algebra homomorphism we also have that $\bra{\vb{x}}_{a}=J^{-1}(\bra{\vb{x}}_{f})$. 

Therefore, we must have that
\begin{equation}
\mel{\vb{y}}{M}{\vb{x}}_{a}=J^{-1}(\bra{\vb{y}})J^{-1}(M')J^{-1}(\ket{\vb{x}})=J^{-1}(\mel{\vb{y}}{M'}{\vb{x}}_{f}),
\end{equation}
where $M'$ is a fermionic Bogoliubov transformation. Now, by the fermionic-anyon commutation relations, we know that
\begin{equation}
J^{-1}(\braket{\vb{y}}{\vb{x}}_{f})=\braket{\vb{y}}{\vb{x}}_{a}=\prod_{i=1}^{n}\delta_{x_{i},y_{i}}=\braket{\vb{y}}{\vb{x}}_{f},
\end{equation}
for any bit-strings $\vb{x}$ and $\vb{y}$. This fact, together with the linearity of $J^{-1}$, implies that 
\begin{equation}
J^{-1}(\mel{\vb{y}}{M'}{\vb{x}}_{f})=\mel{\vb{y}}{M'}{\vb{x}}_{f}.
\end{equation}

Now, since Bogoliubov transformations for fermions are simulated efficiently by classical computers, there must exist a polynomial-time algorithm for computing $\abs{\mel{\vb{y}}{M'}{\vb{x}}_{f}}^{2}$. Since $\abs{\mel{\vb{y}}{M'}{\vb{x}}_{f}}^{2}=\abs{\mel{\vb{y}}{M}{\vb{x}}_{a}}^{2}$, the same polynomial-time algorithm must efficiently simulate the fermionic-anyon Bogoliubov transform $M$, proving our claim.
%\end{proof}

%%%%%
\section{An example with fermionic-anyon beam-splitter}\label{app: example}

For general $N$-modes fermionic systems, a single particle operation is written as a unitary operation in the form 
\begin{align}\label{eq: sp trans fermions}
{f^{\prime}}^{\dagger}_{i}=\sum_{j=1}^{m}U_{ij}f^{\dagger}_{j},
\end{align}
where $U_{ij}$ are the elements of the unitary transformation. For fermions this transformation characterizes a change of basis on the fermions, it is also a canonical transformation as it does not affect the fermionic commutation relation. An example of transformation like Eq.\ \eqref{eq: sp trans fermions} is the fermionic beam splitter
\begin{equation}\label{eq: bs fermions}
{BS}\suptiny{0}{0}{\mathcal{F}}_{m,n}(\theta)=\exp{\left(i\theta(f^{\dagger}_{n}f_{m}+f^{\dagger}_{m}f_{n})\right)}
\end{equation} 
where $m$ and $n$ are mode variables and $\cos{\theta}$ is the transmission amplitude. As shown in \cite{tosta_quantum_2019}, it is possible to define a fermionic anyons version of Eq.\ \eqref{eq: bs fermions}, with the same fashion 
\begin{equation}\label{eq: bs anyons}
{BS}\suptiny{0}{0}{\mathcal{A}}_{m,n}(\theta)=\exp{\left(i\theta(a^{\dagger}_{n}a_{m}+a^{\dagger}_{m}a_{n})\right)}
\end{equation}
where $a^{\dagger}_{n}(a_{m})$ are the anyonic creation(annihilation) operators over mode $n(m)$. By using the commutation relations for fermionic anyons, the expansion in series of the beam-splitter unitary above can be written as 
\begin{equation}\label{eq: anyonic_beam_splitter}
{BS}\suptiny{0}{0}{\mathcal{A}}_{m,n}(\theta)=1+i\sin{\theta}(a^{\dagger}_{n}a_{m}+a^{\dagger}_{m}a_{n})+(\cos{\theta}-1)(a^{\dagger}_{n}a_{m}+a^{\dagger}_{m}a_{n})^{2}. 
\end{equation}
 If we restrict the operations in the computational model for fermionic anyons, given in \cite{tosta_quantum_2019}, to allow only beam-splitters between nearest-neighbour modes we obtain a model with the same computational power to the \textit{Matchgates} model (see \cite{brod_extending_2011,brod_geometries_2012}), which is easy to simulate classically. Therefore we can try to see if the canonical transformations generated by multimode interferometers built using only nearest-neighbour beams-splitters and general phase-shifters preserve the usual notion of a separable state in the anyonic case. To illustrate let us consider a system with four modes two anyons, described by the state 
\begin{align}
    \label{eq: ex initial state}
    \ket{\psi}=\frac{1}{\sqrt{2}}( a_{1}^{\dagger}a_{2}^{\dagger}+ 
   a_{1}^{\dagger}a_{4}^{\dagger})\ket{0}.
\end{align}
Considering the usual notion of separability for identical particles, a two anyons pure state $\ket{\psi}$  is separable if it can be written as a \emph{single Slater permanent}, for any two modes $i$ and $j$,  
\begin{equation}
    \ket{\Psi} = a_i^{\dagger} a_j^{\dagger}\ket{0}. 
\end{equation}

Therefore, considering this definition of a separable state for two anyons pure state, the state in Eq.~\eqref{eq: ex initial state} is separable.  As the beam splitter is a single particle transformation, we expect it preserves the separability of $\ket{\psi}$, in Eq.~\eqref{eq: ex initial state}. Performing the beam splitter in Eq.~\eqref{eq: bs anyons} over modes 1 and 2 on state $\ket{\psi}$, the action results in the state 
\begin{equation}\label{eq: ex output state}
   \ket{\psi\subtiny{0}{0}{\theta}} = \frac{1}{\sqrt{2}}\left( a_{1}^{\dagger}a_{2}^{\dagger}
    + \cos\theta a_{1}^{\dagger}a_{4}^{\dagger} + i\sin\theta a_{2}^{\dagger}a_{4}^{\dagger}\right)\ket{0},
\end{equation}
where $\ket{\psi\subtiny{0}{0}{\theta}}={BS}\suptiny{0}{0}{\mathcal{A}}_{m,n}(\theta)\ket{\psi}$.
In order to check if the output remains separable, one can calculate the von Neumann entropy of the single-particle density matrix. Before calculating the single-particle density matrix, we need to do a digression to explain how to apply the partial trace. A general quantum state of $N$ fermionic anyons can be written  as
\begin{equation}
 \ket{\psi} = \sum_{I_n} w_{I_N} a_{i_1}^\dagger ... a_{i_N}^\dagger \ket{\textrm{vac}}
\end{equation}
where $\ket{\textrm{vac}}$ is the vacuum state, $I_N = (i_1,...,i_N)$ is a shorthand for the list of particle indices, and 
$\sum_{I_N} |w_{I_N}|^2=1$. We can also describe it by its density matrix $ \rho = \ketbra{\psi}{\psi}$, whose elements we write as
\begin{equation}\label{eq.dm.elements}
  \rho(I_N;J_N) = \langle \textrm{vac}| a_{I_N}
  \rho\, \, 
  a_{J_N}^\dagger |\textrm{vac} \rangle,
\end{equation}
where $ a_{I_N}^\dagger =  a_{i_1}^\dagger \ldots  a_{i_N}^\dagger$ and $ a_{J_N}^\dagger =  a_{j_1}^\dagger \ldots  a_{j_N}^\dagger$.

Consider a system composed of $N$ fermionic anyons, and a bipartition of the particles into complementary subsets $A$ and $B$, composed of $N_{A}$ and $N_B$ particles, respectively, such that $A \bigcup \,B = \{1,\ldots,N\}$ and $A \bigcap B = \emptyset$.
The partial trace over the set of particles in $B$ (similarly for $A$) is performed by integration of their corresponding degrees of freedom. 
Using the notation of Eq.~\eqref{eq.dm.elements} we have that
\begin{eqnarray}\label{eq.partial.trace.elementwise}
  \rho_{N_A}(I_{N_A}; J_{N_A}) 
 & = & \sum_{i_k=j_k | k \in B }  \rho(I_{N};J_{N}),
\end{eqnarray}
representing the reduced state of $N_A$ particles. 

Consider a simple case of $N=2$ particles, with $x =  \{1\} $ and $y = \{2\} $. 
Tracing out particle x we obtain
\begin{equation}\label{eq: partial trace x}
 \rho_y(i_2;j_2) = \bra{\textrm{vac}}  a_{i_2} \left( \sum_{i_1}  a_{i_1} \ketbra{\psi}{\psi}  a^\dagger_{i_1} \right)  a^\dagger_{j_2} \ket{\textrm{vac}},
\end{equation}
directly from Eq.~\eqref{eq.dm.elements}. On the other hand, by tracing out particle y we obtain
\begin{equation}
\label{eq: partial trace y}
 \rho_x(i_1;j_1) = \bra{\textrm{vac}}  a_{i_1} \left( \sum_{i_2} e^{i \phi (\epsilon_{i_2 i_1} + \epsilon_{j_1 i_2}) }  a_{i_2} \ketbra{\psi}{\psi}  a^\dagger_{i_2} \right)  a^\dagger_{j_1} \ket{\textrm{vac}}.
\end{equation}

Now, we can apply the partial trace in the state in Eq.~\eqref{eq: ex output state} to verify the separability. The anyonic algebra creates some changes in the notion of inner product in Fock space. There is a dependence in the algebraic phase referent to the spacial mode to be traced out. Therefore, the single-particle density matrix of  $\ket{\psi\subtiny{0}{0}{\theta}} $ carries this dependence explicitly
\begin{align}\label{eq: single particle output}
    \rho\subtiny{0}{0}{sp}=\frac{1}{4}\left(
    \scalemath{0.65}{
\begin{array}{cccc}
 1+ \cos ^2(\theta ) & i \sin (\theta ) \cos (\theta ) & 0 &  i \sin (\theta ) e^{i \phi} \\
 - i  \sin (\theta ) \cos (\theta ) &  1 + \sin ^2(\theta ) & 0 & \cos(\theta ) \\
 0 & 0 & 0 & 0 \\
 - i  e^{-i \phi } \sin (\theta ) &  \cos (\theta ) & 0 & 1  \\
\end{array}}
\right).
\end{align}
The dependence of $\phi$ in the reduced density matrix brings a misleading notion of entanglement in this state, as we can see in Fig.~\ref{fig: entanglement entropy x and y}, where the entanglement was calculated with the von Neumann entropy.
\begin{figure}[t]
    \centering
    \includegraphics[scale=0.3]{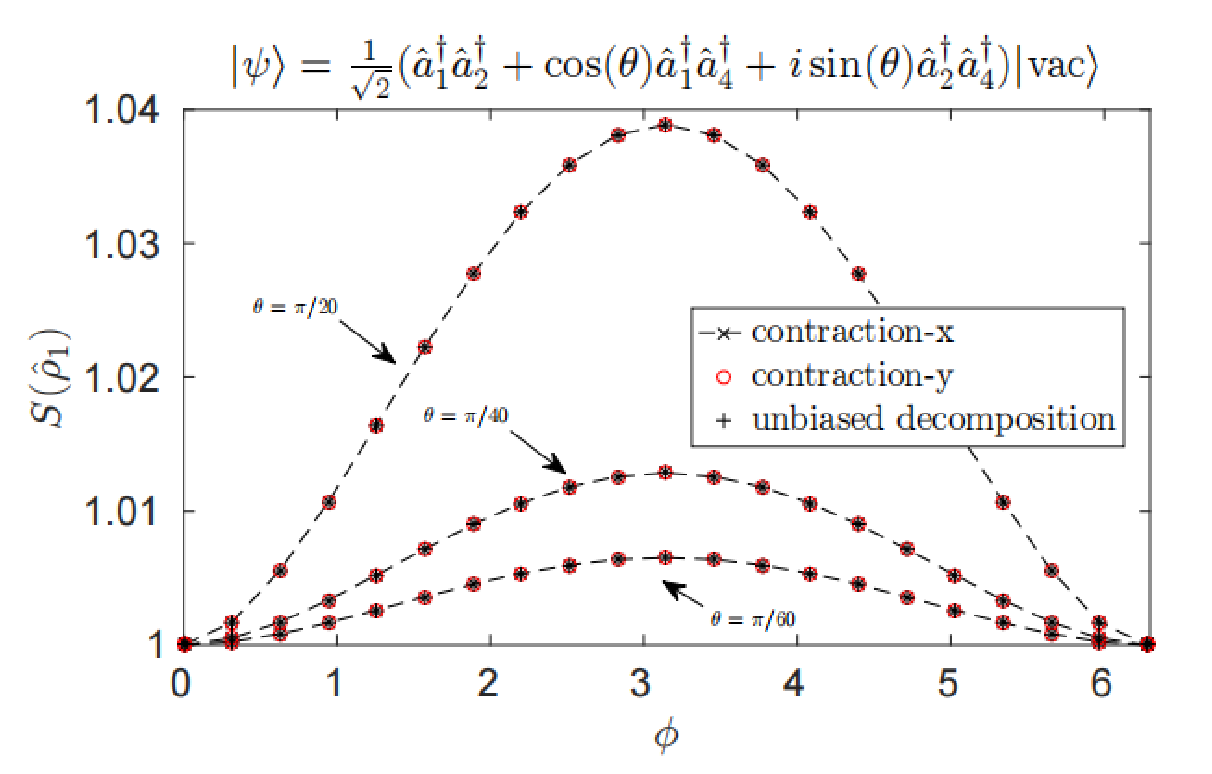}
    \caption{The von Neumann entropy of the single particle state for different values of $\theta$ in the function of $\phi$. The label $x$ indicates partial trace according to Eq.~\eqref{eq: partial trace x}, and the label $y$ indicates partial trace according to Eq.~\eqref{eq: partial trace y}.}
    \label{fig: entanglement entropy x and y}
\end{figure}
The entanglement behaviour varies in function of $\theta$ and $\phi$.
The dependence in $\phi$ arises from the terms $ \rho(1;4)$ and $ \rho(4;1)$, for $$ \rho(1;4)=-i\sin(\theta)\bra{0}a_{1}a_{2}a_{1}^{\dagger}a_{2}^{\dagger}\ketbra{0}{0}a_{4}a_{2}a_{2}^{\dagger}a_{4}^{\dagger}\ket{0}=i\sin(\theta)\exp(i\phi)\bra{0}a_{1}a_{1}^{\dagger}\ketbra{0}{0}a_{4}a_{4}^{\dagger}\ket{0}=i\sin(\theta)\exp(i\phi).$$ Also, similarly, the term $$ \rho(4;1)=i\sin(\theta)\bra{0}a_{4}a_{2}a_{2}^{\dagger}a_{4}^{\dagger}\ketbra{0}{0}a_{2}a_{1}a_{2}^{\dagger}a_{1}^{\dagger}\ket{0}=-i\sin(\theta)\exp(-i\phi)\bra{0}a_{4}a_{4}^{\dagger}\ketbra{0}{0}a_{1}a_{1}^{\dagger}\ket{0}=-i\sin(\theta)\exp(-i\phi)$$ has dependence in $\phi$. 
Actually, the $\phi$ dependence in the coherence of the single-particle density matrices reflects its eigenvalues and consequently the entanglement. This dependence could be suppressed if the single-particle density matrix is already diagonal, after the partial trace, It would analogous to obtaining the Schmidt decomposition for anyonic fermions.  

Considering now, one performs the permutation operation on $\ket{\psi\subtiny{0}{0}{\theta}}$, resulting
\begin{equation}
    \M{P}\ket{\psi\subtiny{0}{0}{\theta}} = \exp(-i\phi) \frac{1}{\sqrt{2}}\left( a_{2}^{\dagger}a_{1}^{\dagger}
    + \cos\theta a_{4}^{\dagger}a_{1}^{\dagger} + i\sin\theta a_{4}^{\dagger}a_{2}^{\dagger}\right)\ket{0},
\end{equation}
it carries a global phase without any implication in the global density matrix. Although, the single particle state gets some local phases, resulting 
\begin{align}\label{eq: single particle output permuted}
    \rho\suptiny{0}{0}{\M{P}}\subtiny{0}{0}{sp}=\frac{1}{4}\left(
    \scalemath{0.65}{
\begin{array}{cccc}
 1+ \cos ^2(\theta ) & i \sin (\theta ) \cos (\theta ) & 0 &  i \sin (\theta ) e^{-i \phi} \\
 - i  \sin (\theta ) \cos (\theta ) &  1 + \sin ^2(\theta ) & 0 & \cos(\theta ) \\
 0 & 0 & 0 & 0 \\
 - i  e^{+i \phi } \sin (\theta ) &  \cos (\theta ) & 0 & 1  \\
\end{array}}
\right),
\end{align}
where $\rho\suptiny{0}{0}{\M{P}}\subtiny{0}{0}{sp} = \sum_{k}a_k \M{P}\ketbra{\psi\subtiny{0}{0}{\theta}}{\psi\subtiny{0}{0}{\theta}}\M{P}^{\dagger}a_k^{\dagger}$.
Besides anyonic permutation relations result on different single-particle states, the von Neumann of the single-particle density matrices are the same, as shown in Fig.~\ref{fig: entanglement entropy x and y}. The action of permutations acts as we are tracing out another particle, as shown in Eq.~\eqref{eq: partial trace x} and Eq.~\eqref{eq: partial trace y} the trace of one particle or another can add a phase due to the permutation of anyons particles. To see this, considering the terms that have a dependence in $\phi$ in the single-particle density matrix, we will apply the partial trace as calculated Eq.~\eqref{eq: partial trace y} for the state Eq.~\eqref{eq: ex output state}. For the element 
\begin{align}
 \rho\subtiny{0}{0}{sp}(1;4)=-i\sin(\theta)\bra{0}a_{1}\exp(i\phi(\epsilon_{21}+\epsilon_{42}))a_{2}a_{1}^{\dagger}a_{2}^{\dagger}\ketbra{0}{0}a_{4}a_{2}a_{2}^{\dagger}a_{4}^{\dagger}\ket{0}=&i\sin(\theta)\exp(-i\phi)\bra{0}a_{1}a_{1}^{\dagger}\ketbra{0}{0}a_{4}a_{4}^{\dagger}\ket{0} \nonumber \\ \nonumber
=&i\sin(\theta)\exp(-i\phi)=  \rho\suptiny{0}{0}{\M{P}}\subtiny{0}{0}{sp}(1;4). 
\end{align} 
And the term 
\begin{align}
 \rho\subtiny{0}{0}{sp}(4;1)=i\sin(\theta)\bra{0}a_{4}\exp(i\phi(\epsilon_{24}+\epsilon_{12}))a_{2}a_{2}^{\dagger}a_{4}^{\dagger}\ketbra{0}{0}a_{2}a_{1}a_{2}^{\dagger}a_{1}^{\dagger}\ket{0}=&-i\sin(\theta)\exp(i\phi)\bra{0}a_{4}a_{4}^{\dagger}\ketbra{0}{0}a_{1}a_{1}^{\dagger}\ket{0} \nonumber \\ \nonumber
=&-i\sin(\theta)\exp(i\phi)= \rho\suptiny{0}{0}{\M{P}}\subtiny{0}{0}{sp}(4;1).
\end{align}
We show that in this way that the terms $ \rho\suptiny{0}{0}{\M{P}}\subtiny{0}{0}{sp}(1;4)$ and $ \rho\suptiny{0}{0}{\M{P}}\subtiny{0}{0}{sp}(4;1)$ of the single-particle density matrix of the state with permutation Eq.~\eqref{eq: single particle output permuted} are obtained tracing out the particle (y) instead of the particle (x) of the state Eq.~\eqref{eq: ex output state}.   

However, applying the fractional Jordan-Wigner transform (FJWT) in the state of the Eq.~\eqref{eq: ex output state}
\begin{equation}\label{eq: ex output state fermions}
   J\ket{\psi\subtiny{0}{0}{\theta}} = \exp(-i\phi)\frac{1}{\sqrt{2}}\left( f_{1}^{\dagger}f_{2}^{\dagger}
    + \cos\theta f_{1}^{\dagger}f_{4}^{\dagger} + i\sin\theta f_{2}^{\dagger}f_{4}^{\dagger}\right)\ket{0},
\end{equation}
%\red{TD. Use sempre a mesma notação usada no restante do texto.} 
and after tracing out one particle we get the following single particle density matrix 
\begin{align}\label{eq: single particle output fermions}
    \rho\subtiny{0}{0}{sp}\suptiny{0}{0}{f}=\frac{1}{4}\left(
    \scalemath{0.65}{
\begin{array}{cccc}
 1 + \cos ^2(\theta ) & i \sin (\theta ) \cos (\theta ) & 0 &  i \sin (\theta ) \\
 - i  \sin (\theta ) \cos (\theta ) &  1 + \sin ^2(\theta ) & 0 & \cos(\theta ) \\
 0 & 0 & 0 & 0 \\
 - i  \sin (\theta ) &  \cos (\theta ) & 0 & 1  \\
\end{array}}
\right),
\end{align}
it is independent of $\phi$. The single-particle density matrix is obtained through partial trace $\rho \suptiny{0}{0}{f} \subtiny{0}{0}{sp} =\sum_{k}\bra{0}f_k \ketbra{\psi\subtiny{0}{0}{\theta}}{\psi\subtiny{0}{0}{\theta}}f_k^{\dagger}\ket{0}$, the partial trace is independent of the particles being trace out, such that for tracing out particle x
\begin{equation}\label{eq: partial trace x fermions}
 \rho_y(i_2;j_2) = \bra{\textrm{vac}} f_{i_2} \left( \sum_{i_1} f_{i_1} \ketbra{\psi}{\psi} f^\dagger_{i_1} \right) f^\dagger_{j_2} \ket{\textrm{vac}}.
\end{equation}
And tracing out particle y we obtain
\begin{equation}
\label{eq: partial trace y fermions}
 \rho_x(i_1;j_1) = \bra{\textrm{vac}}  -f_{i_1} \left( \sum_{i_2} f_{i_2} \ketbra{\psi}{\psi} f^\dagger_{i_2} \right)  -f^\dagger_{j_1} \ket{\textrm{vac}} = \bra{\textrm{vac}}  f_{i_1} \left( \sum_{i_2} f_{i_2} \ketbra{\psi}{\psi} f^\dagger_{i_2} \right)  f^\dagger_{j_1} \ket{\textrm{vac}} =  \rho_y(i_2;j_2).
\end{equation}
 Also, the permutation operation on $J(\mathcal{P}\ket{\psi\subtiny{0}{0}{\theta}})$ only gives a minus sign resulting  
 \begin{equation}\label{eq: ex output state fermions permuted}
   J(\mathcal{P}\ket{\psi\subtiny{0}{0}{\theta}}) = -\exp(-i\phi)\frac{1}{\sqrt{2}}\left( f_{2}^{\dagger}f_{1}^{\dagger}
    + \cos\theta f_{4}^{\dagger}f_{1}^{\dagger} + i\sin\theta f_{4}^{\dagger}f_{2}^{\dagger}\right)\ket{0},
\end{equation}
the minus sign does not change the single-particle density matrix, because does not affect the density matrix of the global state. Furthermore, the partial trace for fermions does not generate phases because it is independent of the ordering of the particles, then the single-particle density matrices are independent of $\phi$ and the permutations. Now, for fermions the term 
$$ \rho\subtiny{0}{0}{sp}\suptiny{0}{0}{f}(1;4)=-i\sin(\theta)\bra{0}f_{1}f_{2}f_{1}^{\dagger}f_{2}^{\dagger}\ketbra{0}{0}f_{4}f_{2}f_{2}^{\dagger}f_{4}^{\dagger}\ket{0}=i\sin(\theta)\bra{0}f_{1}f_{1}^{\dagger}\ketbra{0}{0}f_{4}f_{4}^{\dagger}\ket{0}=i\sin(\theta),$$
and the term $$ \rho\subtiny{0}{0}{sp}\suptiny{0}{0}{f\M{P}}(1;4)=-i\sin(\theta)\bra{0}f_{1}f_{2}f_{2}^{\dagger}f_{1}^{\dagger}\ketbra{0}{0}f_{2}f_{4}f_{2}^{\dagger}f_{4}^{\dagger}\ket{0}=i\sin(\theta)\bra{0}f_{1}f_{1}^{\dagger}\ketbra{0}{0}f_{4}f_{4}^{\dagger}\ket{0}=i\sin(\theta)= \rho\subtiny{0}{0}{sp}\suptiny{0}{0}{f}(1;4).$$
 
 The eigenvalues of the single-particle density matrix in the fermionic representation are $\omega_{\mu}=\{1/2,1/2,0,0\}$, so the state in Eq.~\eqref{eq: single particle output fermions} is not entangled according to Slater decomposition \cite{schliemann_quantum_2001}, since it has Slater rank equals to one. The single particle density matrices of the Eq.\eqref{eq: single particle output} and Eq.~\eqref{eq: single particle output permuted} present a von Neumann entropy dependent on $\phi$ and $\theta$ as shown in Fig.~\ref{fig: entanglement entropy x and y}, while in the fermionic representation, the von Neumann entropy is equal to one independently of values of $\phi$ and $\theta$. Using Theorem \ref{theorem: SM_SD for anyons} we know that the state in the fermionic representation after the application of the Slater decomposition can be mapped, with inverse FJWT, to a state of anyons with the same coefficients. Therefore, the Schmidt coefficients of anyons are equal to $\omega_{\mu}$, independently of $\phi$ and $\theta$, such that the state in Eq.~\eqref{eq: ex output state} is not entangled for all $\phi$ and $\theta$.

\end{document}